\PassOptionsToPackage{
  colorlinks,
  linkcolor=blue,
  filecolor=blue,
  citecolor=blue,
  urlcolor=blue,
  pdfstartview=FitH,
  pagebackref
}{hyperref}
\PassOptionsToPackage{nameinlink}{cleveref}
\documentclass[a4paper,UKenglish,cleveref, autoref, thm-restate, nolineno]{lipics-v2021}

\pdfoutput=1 
\hideLIPIcs  

\usepackage{amssymb}
\usepackage{graphicx}
\usepackage[ruled,linesnumbered,vlined]{algorithm2e}
\usepackage{microtype}
\usepackage{amsmath}
\usepackage{amsthm}
\usepackage{enumerate}
\usepackage{multirow}
\usepackage{nicefrac}
\usepackage[colorinlistoftodos,prependcaption,textsize=small]{todonotes}
\usepackage{multicol}
\usepackage{microtype}
\usepackage{url}
\usepackage{multirow}
\usepackage{algorithm2e}
\usepackage{algorithm2e}
\usepackage{makecell}

\definecolor{forestgreen}{rgb}{0.13, 0.55, 0.13}

\SetCommentSty{mycommfont}

\newtheorem{constraint}{Constraint}

\crefname{constraint}{constraint}{constraints}
\Crefname{constraint}{Constraint}{Constraints}

\graphicspath{{./figures/}}

\def\A{\mathcal{A}}

\def\D{\mathcal{D}}
\def\F{\mathcal{F}}

\def\S{\mathcal{S}}

\def\O{\mathcal{O}}

\def\R{\mathcal{R}}
\def\W{\mathcal{W}}
\def\reals{{\mathbb R}}

\def\eps{{\varepsilon}}

\def\OPT{\text{OPT}}

\def\dfn#1{\emph{\textcolor{forestgreen}{\textbf{#1}}}}

\date{}

\newcommand{\old}[1]{{{}}}

\newcommand{\temp}[1]{{{}}}
\newcommand{\fullversion}[1]{{{}}}

\newcommand{\oldparagraph}[1]{\vspace{5pt}\noindent\textbf{#1}}

\title{Unlabeled Multi-Robot Motion Planning with Improved Separation Trade-offs}

\titlerunning{Unlabeled MRMP with Improved Separation Trade-offs}

\author{Tsuri Farhana} {The Stein Faculty of Computer and Information Science, Ben-Gurion University of the Negev, Beer Sheva, Israel} {tsurif@post.bgu.ac.il} {https://orcid.org/0009-0006-7246-6805}{}

\author{Omrit Filtser} {Department of Mathematics and Computer Science, The Open University of Israel, Ra'anana, Israel} {omrit.filtser@gmail.com} {https://orcid.org/0000-0002-3978-1428}{This research was supported by the ISRAEL SCIENCE FOUNDATION (grant No. 2135/24)}

\author{Shalev Goldshtein} {Department of Mathematics and Computer Science, The Open University of Israel, Ra'anana, Israel} {shalev061@gmail.com} {https://orcid.org/0009-0001-7975-7721}{This research was supported by the ISRAEL SCIENCE FOUNDATION (grant No. 2135/24)}

\authorrunning{T. Farhana, O. Filtser and S. Goldshtein} 

\Copyright{Tsuri Farhana, Omrit Filtser and Shalev Goldshtein}



\category{} 

\relatedversion{A shorter version of this paper appeared in the Proceedings of the 42nd International
Symposium on Computational Geometry (SoCG 2026).} 

\acknowledgements{We thank Dan Halperin for introducing us to the problem and for helpful discussions. We are also grateful to the anonymous reviewers for their insightful comments, which helped improve the structure and presentation of the paper.}

\nolinenumbers

\begin{document}

\maketitle

\begin{abstract}
We study unlabeled multi-robot motion planning for unit-disk robots in a polygonal environment. Although the problem is hard in general, polynomial-time solutions exist under appropriate separation assumptions on start and target positions. Banyassady et al. (SoCG'22) guarantee feasibility in simple polygons under start--start and target--target distances of at least $4$, and start--target distances of at least $3$, but without optimality guarantees. 
Solovey et al. (RSS'15) provide a near-optimal solution in general polygonal domains, under stricter conditions: start/target positions must have pairwise distance at least $4$, and at least $\sqrt{5}\approx2.236$ from obstacles. 
This raises the question of whether polynomial-time algorithms can be obtained in even more densely packed environments.

In this paper we present a generalized algorithm that achieve different trade-offs on the robots-separation ($\rho$) and obstacles-separation ($\omega$), 
all significantly improving upon the state of the art. Specifically, we obtain polynomial-time constant-approximation algorithms to minimize the total path length when (i) $\rho=2\tfrac{2}{3}$ and $\omega=1\tfrac{2}{3}$, or (ii) $\rho\approx3.291$ and $\omega\approx1.354$. 
These solutions are weakly-monotone; we also provide a monotone solution requiring $\omega=\approx1.614$ and $\rho=4$. We prove that monotone plans may not exist when $\omega<1.614$, and weakly-monotone plans may not exist when $\omega<1.354$.
We then present trade-offs between the separation bounds and the approximation factor, specifically achieving an (almost) optimal bound of $\rho=2$ at the cost of a linear approximation factor and requiring $\omega=2$. This applies also for the labeled variant of MRMP, in which case we show a tight bound on $\omega$.
 
Finally, we show that without any robots-separation assumption, obstacles-separation of at least $1.5$ may be necessary for a solution to exist.

\end{abstract}

\newpage
\tableofcontents

\newpage

\section{Introduction}
In the classic multi-robot motion-planning (MRMP) problem, a set of robots operate in a common workspace, and need to reach a set of target positions. More precisely, given a set of $m$ starting positions and a set of $m$ target positions in a workspace $\W$, the goal is to compute a motion strategy for the robots to reach the targets while avoiding collisions with the boundary of the workspace and with other robots. In this paper we consider unit-disk robots that are moving in a polygonal workspace in the plane. If each robot is assigned to a specific target, the problem is referred to as \emph{labeled} MRMP, and otherwise it is \emph{unlabeled} MRMP. Unless otherwise specified, the variant of MRMP that we consider in this paper is unlabeled.

The MRMP problem and its variants have been widely investigated, both from theoretical and practical perspectives; for some of the more recent results, see, for example,~\cite{ABHSS25,AGHT23,AHSS24,AHS25,EHRSS23,FKKNRS24}.
Very recently, both the labeled and unlabeled variants of MRMP for unit disks in polygons with holes were shown to be PSPACE-hard~\cite{ABBKLS25}. Many other variants of the problem are known to be hard as well, see e.g.~\cite{BHKLS21,HSS84,SH15,SY84}.

Interestingly, the separation between the robots plays a key role in the difficulty of the problem. Adler, de Berg, Halperin, and Solovey~\cite{ABHS15} show that in a simple polygonal workspace with $n$ vertices and $m$ unit-disk robots, when the \dfn{robots-separation}, denoted $\rho$, is $4$ (i.e., all the start and target positions are at least distance $4$ apart), a solution for unlabeled MRMP always exists, and can be found in $\Tilde{O}(mn+m^2)$ time.
In~\cite{BBBBFHKOS22} this result is further improved by distinguishing between two types of robots-separation bounds: (i) \dfn{monochromatic robots-separation} --- the distance between two start positions or between two target positions, and (ii) \dfn{bichromatic robots-separation} --- the distance between a start and a target positions. They present an algorithm with the same running time of $O(n\log n+mn+m^2)$, but under weaker assumptions: monochromatic separation of $4$, and bichromatic separation of $3$. Moreover, they prove that these bounds are tight, i.e., a solution for unlabeled MRMP in a simple polygon may not always exist when the monochromatic separation is $4-\eps$, or when the bichromatic separation is $3-\eps$. In addition, they show that when the free space consists of a single connected component, the bichromatic separation is not required for a solution to always exist.
The algorithmic approach in both papers \cite{BBBBFHKOS22, ABHS15} is similar: restricting the robots to move one at a time (a \dfn{monotone} strategy) on a motion graph that has the start and target positions as vertices. A drawback of this approach is that it does not provide any reasonable approximation on the length of the solution, i.e., the total length of the paths in the motion plan.

Solovey, Yu, Zamir, and Halperin~\cite{SYZH15} present a different approach for unlabeled MRMP, with optimality guarantee even in polygons with holes, but at the cost of adding an additional assumption on the distance between the obstacles to any start and target position. Specifically, they consider the case when the robots-separation is $\rho=4$, and the \dfn{obstacles-separation}, denoted $\omega$, is $\sqrt{5}\approx2.236$ (i.e., the distance between any start/target position and any obstacle is at least $\sqrt{5}$).
Under these assumptions, they present an algorithm that runs in $\Tilde{O}(m^4+n^2m^2)$ time and returns a solution of total length at most $\OPT+4m$. Here and throughout the paper, $\OPT$ denotes the minimum possible sum of lengths of the paths in a collision-free motion plan.

Notice that when assuming a large robots-separation ($\rho\ge4$), the overall ``density'' of the robots inside the workspace becomes small, which may require the workspace to be much larger than the number of robots operating in it.
Also, assuming a large obstacles-separation requires adding an additional buffering ``layer'' around the obstacles. 
This motivates the question of whether we can find weaker assumptions on either the robots-separation or the obstacles-separation, that are sufficient for having a polynomial time algorithm for unlabeled MRMP.
Surprisingly, in this paper we show that both the robots-separation and obstacles-separation assumptions can be significantly relaxed, while still enabling a polynomial-time algorithm that achieves a constant-factor approximation of the solution length.

In previous works, namely \cite{ABHS15,BBBBFHKOS22,SYZH15}, the algorithmic approach uses a \dfn{monotone} strategy, where robots move to their target position one by one according to some order, while other robots must remain at their start/target positions. 
The main result in our paper uses a weakly-monotone strategy.
In a \dfn{weakly-monotone} motion plan, each start/target position $p$ is associated with a revolving area - a disk of radius $r>1$ that contains the unit disk centered at $p$ and is empty of obstacles and other unit disks centered at other start/target positions. The robots move to their target position one by one according to some order, while other robots must remain inside the revolving area associated with their current position. 
Such a weakly-monotone strategy was used, e.g., in~\cite{AGHT23,SH18} for a variant of labeled MRMP with unit-disk robots.

\oldparagraph{Our contribution.}
In this paper we present an algorithm for unlabeled MRMP, that provides tighter trade-offs between robot-separation and obstacles-separation assumptions.

In particular, for a parameter $0\leq\eps\leq1$, a polygonal environment of complexity $n$, and $m$ unit-disc robots, our algorithm computes in $\Tilde{O}(m^4+n^2m^2)$ time a collision-free motion plan of total length $O(\OPT)$, assuming that $\rho=\max\{4-2\eps,2+\eps\}$ and $\omega=\max\{1+\eps,\sqrt{(3-\sqrt{3}-\eps)^2+1}\}$ (see \Cref{thm:opt_sep_bounds}). In \Cref{tab:separation} we highlight three interesting values of $\eps$ that minimizes either $\rho,\omega$ or the guaranteed length of the solution. 
This algorithm can be viewed as a generalization of the algorithm in~\cite{SYZH15}, and indeed, as a consequence of our analysis we positively resolve a conjecture made in~\cite{SYZH15}, by showing that $\omega=1.614$ is sufficient when $\rho=4$.

Moreover, we show that a monotone motion plan (such as the one produced by the algorithm in~\cite{SYZH15}) may not always exist when $\omega<1.614$, and therefore to achieve smaller separation bounds one must use a different strategy.
Indeed, as mentioned above, our motion plan is weakly-monotone. Interestingly, we also show that a weakly-monotone motion plan may not always exist when $\omega<1.354$.

Note that the bounds on $\rho$ mentioned above do not cover the case of $\rho=2$, which is the optimal (monochromatic) robots-separation.
We thus further generalize one part of our algorithm, and show how to trade robots-separation for optimality (and obstacle separation). Specifically, we show that any value of $\rho$ between $2\frac23$ and $2$ can be achieved, where $\rho=2$ requires having $\omega=2$, and the approximation factor becomes linear (see also \Cref{tab:separation2}).

Next, we show that our algorithm implies that a solution to Labeld MRMP always exists when $\rho=2$, $\omega=2$, and every start-target pair are in the same connected component of the free space. We also provide an $O(m)$ approximation for the length of such motion plan. 

Finally, we present lower bounds on the amount of obstacles-separation that is required in order for a solution to always exist. In the unlabeled case, we show that when no assumptions are made on $\rho$ (i.e., $\rho=2$), obstacles-separation of at least $1.5$ is sometimes needed. In the labeled case, we observe that obstacles-separation of at least $2$ is sometimes necessary, and therefor showing that $\omega=2$ is tight in the labeled case.

\begin{table}[h!]
\centering
\begin{tabular}{|l|l|l|l|l|l|}
\hline
$\boldsymbol{\omega}$ &
$\boldsymbol{\rho}$ &
\textbf{Solution length} &
\textbf{Running time} &
\textbf{Notes} &
\textbf{Reference} \\ \hline\hline

$1$ &
\multirow{3}{*}{$4$} &
no guarantee &
$\widetilde{O}(mn+m^2)$ &
Simple polygons &
\cite{ABHS15,BBBBFHKOS22} \\
\cline{1-1}\cline{3-6}

$\sqrt{5}$ &
&
\multirow{2}{*}{$OPT+4m$} &
\multirow{6}{*}{$\widetilde{O}(m^4+n^2m^2)$} &
&
\cite{SYZH15} \\
\cline{1-1}\cline{5-6}

$\approx 1.614$ &
&
&
&
&
\multirow{3}{*}{\Cref{thm:opt_sep_bounds}} \\
\cline{1-3}\cline{5-5}

$\approx 1.354$ &
$\approx 3.291$ &
\multirow{2}{*}{$O(OPT)$} &
&
&
\\
\cline{1-2}\cline{5-5}

$1\frac{2}{3}$ &
$2\frac{2}{3}$ &
&
&
&
\\
\cline{1-3}\cline{5-6}

$1+\eps$ &
$4-2\eps$ &
$O\left(\left(\frac{1}{1-\eps}\right)^2 \cdot OPT\right)$ &
&
$\frac{2}{3}<\eps<1$ &
\Cref{thm:opt_sep_bounds_improved} \\
\cline{1-3}\cline{5-6}

$2$ &
$2$ &
$O(m\cdot OPT)$ &
&
&
\Cref{cor:opt_sep_bounds_eps1} \\ \hline
\end{tabular}

\caption{Summary of algorithms and their separation assumptions.\\
In \cite{BBBBFHKOS22}, the start-target (bichromatic) separation is $3$.}
\label{tab:separation}
\end{table}

\subsection*{Technical ideas and overview of the paper}
Our approach in the first algorithm generalizes the strategy of Solovey et al.~\cite{SYZH15} for unlabeled MRMP, which is briefly described as follows. Given a set $S$ of start positions, a set $T$ of target positions, and a polygonal workspace $\W$, first compute an optimal set $\Gamma$ of paths assigning each starting point to a target point (a perfect matching). 
This set of paths minimizes the sum of geodesic distances, when ignoring the other robots in the workspace.
Then, one can show that when assuming $\rho=4$ and $\omega=\sqrt5$, there is always at least one target which is a so-called standalone goal, that is, a target $t$ that does not block any other path in $\Gamma$. Such a target is useful because we can place a robot there and then treat it as an obstacle.
However, there might be some other robot that blocks the path that leads to $t$, in which case we can make a switch: we push the last robot that blocks this path to $t$.
This switching process is possible because of the assumption $\omega=\sqrt5$. Then, run the algorithm recursively, but after removing the start and target positions that where satisfied in the previous recursive step, and updating the workspace to treat the robot at $t$ as an obstacle.
 
Our weakly-monotone motion plan is obtained by inserting a very important modification to the algorithm of~\cite{SYZH15}: instead of choosing a standalone target in each iteration, we relax the requirement and require only an ``almost standalone'' target --- a target that may interrupt other paths, but only by a small amount. 
In our strategy, we distinguish between ``blocking'' and ``interrupting'' positions. Instead of always switching paths when another robot is standing in the way, we allow robots that interrupt a path by only a small amount to move away from the path. A robot that interrupts a given path may move slightly in its close neighborhood (a disk of radius equal to the amount of interruption plus one) in order to clear the way. In contrast, a robot that blocks a path may not have enough wiggle room, and we will switch the paths.

As suggested above, this combined strategy can be viewed as a generalization of~\cite{SYZH15}:
we introduce another parameter, called the ``overlap parameter'', which is used to determine whether a position is blocking a path or only interrupting it. This parameter allows us to define a set of constraints on $\rho$ and $\omega$, that allows (i) the existence of an interrupting target, (ii) the option to switch paths in the case when a robot is blocking a path, and (iii) the existence of a sufficiently large free neighborhood for an interrupting robot to clear the way.
It turns out that there is an interesting trade-off between these constraints, leading to the main result of our paper (\Cref{thm:opt_sep_bounds}).

Note that the concept of interrupting positions raises some issues that were not present in the algorithm of~\cite{SYZH15}. For example, in each iteration of the recursive algorithm we need to take into account that the robot placed on a target may still interrupt other paths. In~\cite{SYZH15} this is solved simply by removing a unit disk around the interrupting target from the workspace. However, in our case, such a robot may still need to move in order to clear the way for other robots. Moreover, if we simply treat it as an obstacle, the assignment of paths in the next iteration may be too large. We solve this issue by observing that it is actually enough to remove a smaller disk around a settled target when computing the paths, while still allowing it to move within its small neighborhood when planning the motion in the following iterations. Another issue is that there may be two or more robots that interrupt a path and may need to move simultaneously. Here we describe a simple motion with respect to the ray connecting the robot moving on the path with the interrupting robots.

In \Cref{sec:constraints}, we systematically define and analyze the constraints that are required for each part of this strategy, each constraint depends on the overlap parameter. In \Cref{sec:algorithm} we present the complete generalized algorithm and prove its correctness. Then, we show how to choose an overlap parameter and adjust constraints to get the minimum possible obstacles-separation and the minimum possible robots-separation. 

In the first version of our algorithm, to clear a path, we only move robots that intersect that path. Such a strategy cannot work when the robots-separation is exactly $2$, because a robot that intersect a path may be block by other robots that do not intersect that path, and thus will not be able to move and clear the path. Therefore, to achieve a separation bound of $2$ between the robots, we present in \Cref{sec:improved} a modification of the path clearance strategy, in which we also move robots that do not intersect the path. For any $\frac23<\eps\le1$, this strategy allows us to assume a robots-separation of $\rho=4-2\eps$, with obstacles-separation of $\omega=1+\eps$. Since more robots are moving, the approximation factor increases, and for $\eps=1$ we obtain a solution of size $O(m\cdot\OPT)$. We also show how to adapt this strategy for labeled MRMP, and obtain separation bounds of $\rho=2$ and $\omega=2$ in this case (see \Cref{sec:labeled}). 

Note that the monochromatic separation must be at least $2$, so the robots-separation is almost optimal (the bichromatic separation can be $0$). In \Cref{sec:weakly_monotone_limitations} we discuss the limitation of monotone and weakly monotone motion plans, and in \Cref{sec:lower_bounds} and we preset lower bounds on the obstacles-separation. We leave open the question of what is the upper bound on $\omega$ when robots-separation is optimal.

\begin{remark}[Previous versions of this paper]
    In previous versions of this paper (published in SoCG'26 and on ArXiv), we included the Exodus algorithm which achieves bounds of $\rho=2$ and $\omega=3$ in simple polygons, with approximation ratio of $\OPT+4m^2$ and running time of $O\left(m^3+mn^2\right)$. We decided to omit this algorithm from this version of the paper because it is outperformed by the new modification of our first algorithm, which is introduced in this version for the first time.
\end{remark}

\subsection*{More related work}
For labeled MRMP, Agarwal, Geft, Halperin, and Taylor~\cite{AGHT23} (following Solomon and Halperin~\cite{SH18}) consider unit disk robots with revolving areas: each start and each target position is contained in a disk of radius $2$ lying in the workspace, not necessarily concentric with the start or target position, which is free from other start or final positions. Interestingly, although under this assumption we may have instances with $\mu=2$, $\beta=0$, and $\omega=1$, the overall density of the robots inside the workspace is still small, which allow for a solution to always exists.
As mentioned above, Agarwal et al.~\cite{AGHT23} present a motion strategy which is weakly-monotone: robots move to their target according to some ordering, one at a time, while other robots may only move in their revolving areas. The total length of the solution returned by their algorithm is a constant approximation of the optimal strategy.

Other related works include MRMP for a constant number of robots, where polynomial time exact or approximated solutions are known (see, e.g., \cite{ABHSS25,AHS25,AHSS24,EHRSS23}), and a line of work on motion planning in discrete domains, where usually the goal is to optimize the makespan (see, e.g., \cite{DFKSM18,FKKNRS24,Yu18,YL16}).

\section{Preliminaries}
We consider a set of $m$ unit-disk robots, moving in a polygonal \dfn{workspace} $\W\subseteq\reals^2$ (which may be simple or cluttered with polygonal obstacles), whose overall number of edges is $n$. Below, we mostly follow the notations in \cite{BBBBFHKOS22} and \cite{SYZH15}.

The \dfn{obstacle space} $\O$ is the complement of the workspace $\W$, that is, $\O = \reals^2\setminus \W$.
We refer to the points in $\W$ as positions, and say that a robot is at position $x\in \W$ when its center is positioned at the point $x$. For a point $x\in \reals^2$ and a radius $r \in\reals_+$, we define $D_r(x)$ to be the open disk of radius $r$ centered at $x$. The unit-disk robots are defined to be open sets, and therefore two robots collide if and only if the distance between their positions is strictly less than $2$. 
In addition, a robot collides with the obstacle space $\O$ if and only if its center is at distance strictly less than $1$ from $\O$. The set of all positions in $\W$ where a unit-disk robot does not collide with the obstacle space is called the \dfn{free space}, denoted $\F = \{x\in \reals^2\mid D_1(x)\cap \O = \emptyset\}$. 
Note that the free space is a closed set, and denote by $\partial \F$ the boundary of $\F$.

In the multi-robot motion-planning problem, we are given a workspace $\W$ with $n$ edges and two sets $S=\{s_1,\dots,s_m\}$ and $T=\{t_1,\dots,t_m\}$ of points in the free space $\F$. The goal is to plan a collision-free motion for $m$ robots, each positioned on a different starting point in $S$, such that by the end of the motion every target position in $T$ is occupied by some robot.
When the robots are distinguishable (i.e. labeled), the robot that starts at position $s_i$ is required to end up at the target $t_i$. Otherwise, when the robots are indistinguishable (i.e., unlabeled), it does not matter which robot ends up at which target position, as long as all targets are occupied at the end of the motion.

Formally, we wish to find a set $\Gamma$ of continuous paths  $\gamma_1,\dots,\gamma_m$, where $\gamma_i:[0, 1]\rightarrow \F$ for every $1 \le i \le m$, and such that $\bigcup_{i=1}^m\gamma_i(0) = S$ and $\bigcup_{i=1}^m\gamma_i(1) = T$. In addition, the paths need to be collision-free, i.e., at any point in time $t\in [0, 1]$, there are no $i,j$ such that $\gamma_i(t)$ and $\gamma_j(t)$ are at distance strictly less than $2$.
Denote the size of a solution $\Gamma$ by $|\Gamma|=\sum_{\gamma_i\in\Gamma}|\gamma_i|$, where $|\gamma|$ is the length of a path $\gamma$ in the $L_2$ norm.
In the optimization version of the problem, we are interested in finding a solution $\Gamma$ which minimizes $|\Gamma|$, i.e., minimizing the sum of lengths of the paths.

\oldparagraph{Separation assumptions.} 
We denote by $\omega$ the \dfn{obstacles-separation} assumption on an instance of MRMP, i.e., we assume that for every $p\in S\cup T$ and every point $x\in\O$, it holds that $\|x-p\|\ge \omega$ (equivalently, $D_\omega(p)\subset\W$). We denote by $\rho$ the \dfn{robots-separation} assumption, i.e., we assume that for every $p_1,p_2\in S\cup T$ it holds that $\|p_1-p_2\|\ge \rho$.

\subsection*{Overlapping, blocking and interrupting positions}
Let $p,q\in \F$ be two points in the free space. We define the \dfn{overlap} of $p$ and $q$ to be $\max\{0,2-\|p-q\|\}$ (see \Cref{fig:overlap}). If $\D_1(p)\cap\D_1(q)=\emptyset$, then their overlap is $0$. Otherwise, we say that $p,q$ are \dfn{$\eps$-overlapping} for $\eps=2-\|p-q\|$, and note that $\|p-q\|=2-\eps$.
\begin{figure}[h!]
    \centering
    \includegraphics[page=1,scale=0.7]{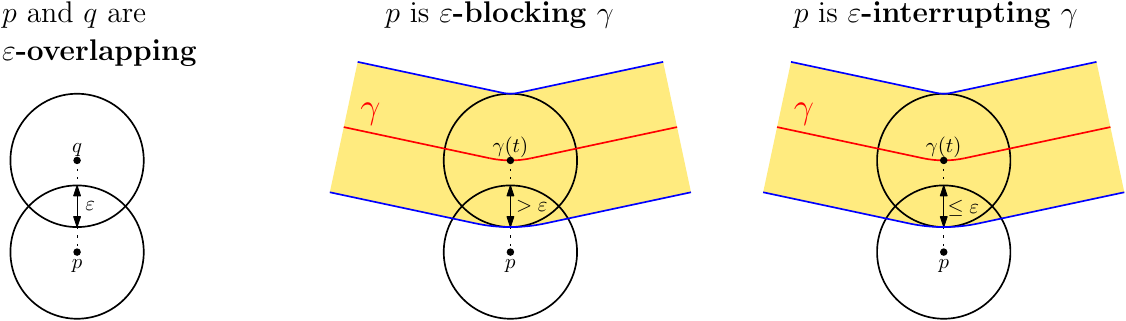}
    \caption{Overlapping, blocking, and interrupting positions.}\label{fig:overlap}
\end{figure}

We say that a position $p\in\F$ is \dfn{$\eps$-blocking} a path $\gamma$ if there exists $0\le t\le 1$ such that the overlap of $p$ and $\gamma(t)$ is \textbf{strictly larger than $\eps$}. Blocking robots (i.e. robots that are placed on blocking positions) may be impossible to clear from the path, and we will have to make a ``switch''.
We say that a point $p\in\F$ is \dfn{$\eps$-interrupting} a path $\gamma$ if for every $0\le t\le 1$, the overlap of $p$ and $\gamma(t)$ is \textbf{at most $\eps$} (that is, it is not $\eps$-blocking).
Interrupting robots (i.e. robots that are placed on interrupting positions) will have enough wiggle room to move away from $\gamma$ and ``clear'' the path.
Notice that if for every $0\le t\le 1$ the overlap of $p$ and $\gamma(t)$ is $0$, then for every $0\le t\le 1$ it holds that $\|\gamma(t)-p\|> 2$, and therefore in this case we say that $p$ is $0$-interrupting $\gamma$.
We say that a path $\gamma$ is \dfn{$\eps$-blocked} if there exists $s\in S$ that $\eps$-blocks it, and that $\gamma$ is \dfn{$\eps$-interrupted} if no $s\in S$ is $\eps$-blocking it.

\oldparagraph{The overlap parameter.} For the analysis in \Cref{sec:constraints}, we fix a parameter $0\le\eps\le 1$ which we call the \dfn{overlap parameter}. In these sections, for simplicity of the presentation, we sometimes omit $\eps$ from the above notions and just say ``blocking'' or ``interrupting''.
In our analysis, we introduce a set of constrains described as functions of $\eps$, which are required for the different parts of our algorithm to work.
In \Cref{sec:optimzing_bounds} we show how to choose $\eps$ in order to optimize different separation bounds.

\section{Constraints for Tighter Separation Bounds} \label{sec:constraints}

To give an intuition for the required constraints, we first provide a high-level overview of our strategy. The algorithm is recursive, and its input is the overlap parameter $\eps$, a set $S$ of starting positions, a set $T$ of target positions, and a workspace $\W$. In each recursive iteration, one robot is moving to a target position, then $S$,$T$ and $\W$ are updated accordingly. 

The first step is to compute an \dfn{optimal-assignment path set}, as it is defined in \cite{SYZH15}:
\begin{definition}
Let $\Gamma$ be a set of $m$ geodesic paths $\gamma_1,\dots,\gamma_m$, where $\gamma_i:[0, 1]\rightarrow \F$ for every $1 \le i \le m$, and such that $\bigcup_{i=1}^m\gamma_i(0) = S$, $\bigcup_{i=1}^m\gamma_i(1) = T$. We say that $\Gamma$ is an optimal-assignment path set for $S,T,\W$ if it minimizes $\sum_{i=1}^m|\gamma_i|$ over all such path sets.

\end{definition}
Note that an optimal assignment path set $\Gamma$ is not necessarily a feasible solution, because it ignores possible collisions between robots. This also means that the sum of paths in $\Gamma$ is at most the length of an optimal solution to our MRMP problem.
Let $\Gamma=\{\gamma_1,\dots,\gamma_m\}$ be an optimal-assignment path set for $S,T,\W$.

In \Cref{sec:geodesics} we present a useful lemma showing that a blocking robot can enter the geodesic path that it blocks. This lemma will be useful throughout the analysis, and it requires one constraint on the obstacles-separation, and another on the robots-separation.

The next step is to find a target $t\in T$ which does not block any path in $\Gamma$ (other than the path $\gamma$ from $s\in S$ to $t$). Then we can move a robot from $s$ to $t$ along $\gamma$, and it will not block other paths in the following iterations. Note that $t$ may still interrupt other paths. In \Cref{sec:almost_standalone} we show that such a target always exists, when assuming an additional constraint on the robots-separation. 

If there is a position $s'\in S$, which blocks the path $\gamma$, we need to make a switch, i.e., we construct a new path in $\F$ from $s'$ to $t$, which is not blocked by any other robot. In \Cref{sec:switch_path} we show how to construct such a switch path $\gamma'$, under all the constraints defined earlier. If there is no such blocking position $s'$, we simply set $\gamma'\gets\gamma$.

We now have a path $\gamma'$ from some $s''\in S$ (which may be either $s$ or $s'$) to the target $t$, which in not blocked by any robots, but may be interrupted. 
In \Cref{sec:clearance_path} we show how to move a robot from $s''$ to $t$ along $\gamma'$, while all interrupting robots (either on start or target positions) move slightly in order to clear the path. 
For this we introduce two additional constraints, one on the obstacles-separation, and one on the robots-separation.

Finally, we prepare the input for the next recursive iteration of the algorithm, by removing $s''$ from $S$, $t$ from $T$, and $D_{1-\eps}(t)$ from $\W$. Removing $D_{1-\eps}(t)$ from $\W$ ensures that the robots placed on $t$ does not block any path that will be computed in any of the next iterations.

In \Cref{sec:algorithm} we present the complete algorithm based on the building blocks presented in this section, prove its correctness and analyze its running time and approximation factor.

\subsection{Geodesics and blocking positions in the free space}\label{sec:geodesics}
We begin by introducing the constraints under which a blocking robot positioned on some $p\in S\cup T$ can enter a geodesic path $\gamma: [0, 1]\rightarrow \F$ between some start and target positions, via a straight line segment connecting it to the path.

\begin{figure}[h!]
    \centering
    \includegraphics[page=3,scale=0.9]{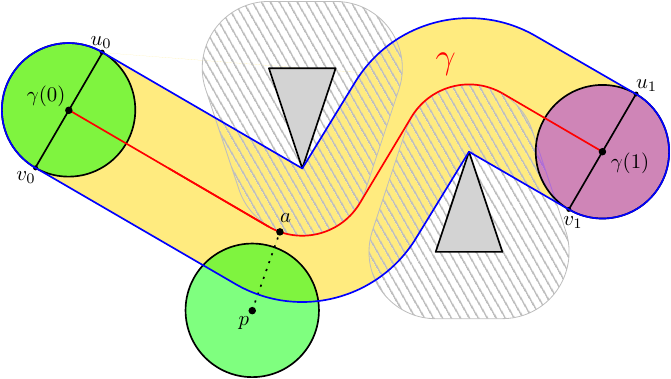}
    \caption{A geodesic path $\gamma$ in $\F$ from a start position $\gamma(0)\in S$ to a target position $\gamma(1)\in T$ is illustrated in red. The obstacles are in gray, and the trace of $\gamma$ is highlighted in yellow. A position $p$ is blocking/interrupting $\gamma$. The point $a$ is the closest point to $p$ on $\gamma$.}
    \label{fig:geodesic_path}
\end{figure}

A geodesic path between two positions $x,y\in \F$ is a shortest length path $\gamma : [0,1] \to \F$ where $\gamma(0)=x$ and $\gamma(1)=y$.
Denote by $\D_1(\gamma)=\bigcup_{0\le t\le 1}\D_1(\gamma(t))$ the \dfn{trace} of $\gamma$ (see \Cref{fig:geodesic_path}). For $0\le t_1\le t_2\le 1$, denote by $\gamma[t_1,t_2]$ the subpath of $\gamma$ between $\gamma(t_1)$ and $\gamma(t_2)$.

In \Cref{apx:geodesics} we provide some very useful geometric properties of geodesic paths in the free space, which result in the following constraints. Note that we have one constraint on the obstacles-separation, and another on the robots-separation. 
\begin{restatable}{constraint}{constraintBoundary}\label{con:boundary_seperation_to_clear_R}  
    $\omega\ge \sqrt{(3-\sqrt{3}-\eps)^2+1}$
\end{restatable}
\begin{restatable}{constraint}{constraintRobots}\label{con:robot_seperation_to_ensure_R}
    $\rho\ge \sqrt{\frac{1}{2}^{2} + \left(3-\frac{\sqrt{3}}{2}-\eps\right)^2}$
\end{restatable}

Let $\gamma: [0, 1]\rightarrow \F$ be a geodesic path. We say that a point $a$ on $\gamma$ is \dfn{locally closest} to $p$ if there exists $\delta>0$ such that $D_{\|p-a\|}(p)\cap \gamma[a-\delta,a+\delta]=a$. 

\begin{restatable}{lemma}{boundarySeparationRho}\label{lem:boundary_seperation_rho}
    Let $\gamma: [0, 1]\rightarrow \F$ be a geodesic path such that $\gamma(0),\gamma(1)\in S\cup T$, and let $p\in S\cup T$ be a position which is $\eps$-blocking $\gamma$. Let $a=\gamma(t)$ be 
    a point on $\gamma$ which is locally closest to $p$ and such that $\|p-a\|<2-\eps$.
    If \Cref{con:boundary_seperation_to_clear_R,con:robot_seperation_to_ensure_R} hold, then 
    $\overline{pa}\subset \F$.
\end{restatable}

Intuitively, \Cref{lem:boundary_seperation_rho} states that a robot blocking $\gamma$ and positioned on $p$ can move to $a$ along the path $\overline{pa}$, as long as no other robots are standing in its way (see \Cref{fig:geodesic_path}).

For clarity of presentation and to avoid interrupting the flow of the exposition, the full proof of \Cref{lem:boundary_seperation_rho} is deferred to \Cref{apx:geodesics}.
Note that in the following subsections we introduce two more constraints on $\rho$ that together make \Cref{con:robot_seperation_to_ensure_R} redundant, however, \Cref{lem:boundary_seperation_rho} holds already for this weaker constraint and may be of independent interest. 

\subsection{Existence of an $\eps$-interrupting target}\label{sec:almost_standalone}
We now define formally a generalization of the standalone goal defined in \cite{SYZH15} (Definition 3). Our definition depends on the overlap parameter $\eps$, and it is equivalent to a standalone goal for $\eps=0$.
Let $\Gamma=\{\gamma_1,\dots,\gamma_m\}$ be an optimal-assignment path set for $S,T,\W$.
\begin{definition} 
A target position $t\in T$ is an \dfn{$\eps$-interrupting target}, if for every path $\gamma\in \Gamma$ such that $\gamma(1)\neq t$, it holds that $t$ is not $\eps$-blocking $\gamma$.
\end{definition}

To show that an $\eps$-interrupting target always exists, we need an additional constraint.
\begin{constraint}\label{con:robot_separation_4-2eps}
        $\rho\ge4-2\eps$
\end{constraint}
The following theorem is a generalization of Theorem 4 in \cite{SYZH15}.
The proof follows a similar logic, however, we need to take into account some issues that arise from the generalized definition and \Cref{con:boundary_seperation_to_clear_R,con:robot_seperation_to_ensure_R,con:robot_separation_4-2eps}.

\begin{restatable}{theorem}{almostStandalone}\label{theorm:almost-standalone}
    Let $\Gamma$ be an optimal-assignment path set for $S,T,\W$. If \Cref{con:boundary_seperation_to_clear_R,con:robot_seperation_to_ensure_R,con:robot_separation_4-2eps} hold, then there exists an $\eps$-interrupting target.
\end{restatable}
\begin{proof}
    Assume by contradiction that there is no $\eps$-interrupting target, that is, for every $t_i\in T$ there exists $\gamma_k\in\Gamma$, $k\ne i$, such that $t_i$ is $\eps$-blocking $\gamma_k$.
    We show that this results in a circular blocking, which leads to a path set of shorter length, in contradiction to the optimality of $\Gamma$. 
    
    Consider a directed graph that has a vertex $v_i$ for each $t_i\in T$, and for every $1\le i,j\le m$, $i\neq j$, if $t_i$ is $\eps$-blocking $\gamma_j$ we add the edge $(v_i,v_j)$ to the graph. Because we assume that there is no $\eps$-interrupting target, every vertex has in-degree at least $1$, and thus the graph must contain a cycle. We denote this cycle by $(v_1,v_2,\dots,v_\ell)$, and assume w.l.o.g. that for every $1\le i\le \ell$, the path $\gamma_i$ is from $s_i$ to $t_i$, i.e., $\gamma_i(0)=s_i$ and $\gamma_i(1)=t_i$.

\begin{figure}[h!]
    \centering
    \includegraphics[page=1,scale=0.8]{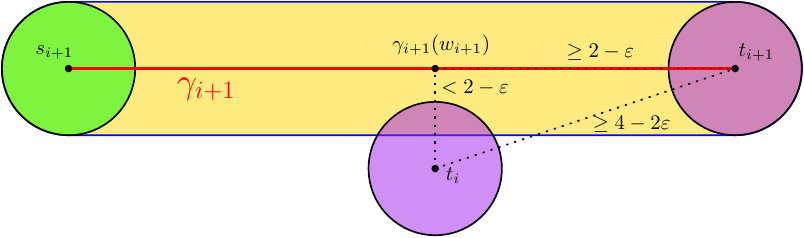}
    \caption{The path $\gamma_{i+1}'$ from $s_{i+1}$ to $t_i$ is shorter then $\gamma_{i+1}$.}
    \label{fig:tringle_inequlity}
\end{figure}

    Consider an edge $(v_i,v_{i+1})$ of the cycle, which corresponds to $\gamma_{i+1}$ being $\eps$-blocked by $t_i$. Let $\gamma_{i+1}(w_{i+1})$ for $w_{i+1}\in [0,1]$ be the point on $\gamma_{i+1}$ closest to $t_i$, then the overlap of $\gamma_{i+1}(w_{i+1})$ and $t_i$ is strictly larger than $\eps$, and thus $\|t_i-\gamma_{i+1}(w_{i+1})\|<2-\eps$. Let $\gamma'_{i+1}$ be the path from $s_{i+1}$ to $t_i$ that follows $\gamma_{i+1}$ from $s_{i+1}$ until reaching $\gamma_{i+1}(w_{i+1})$, and then continue on the straight line segment from $\gamma_{i+1}(w_{i+1})$ to $t_i$ (see \Cref{fig:tringle_inequlity}). 
    Because \Cref{con:boundary_seperation_to_clear_R,con:robot_seperation_to_ensure_R} hold, we can apply \Cref{lem:boundary_seperation_rho} on $\gamma_{i+1}$, $t_i$ and $\gamma_{i+1}(w_{i+1})$, and get that $\overline{\gamma_{i+1}\left(w_{i+1}\right)t_i}\subseteq\F$, and therefore $\gamma_{i+1}'\subseteq\F$.
    By \Cref{con:robot_separation_4-2eps}, $\|t_i-t_{i+1}\|\ge 4-2\eps$, and thus by the triangle inequality we have 
    $$\|\gamma_{i+1}(w_{i+1})-t_{i+1}\|\ge \|t_i-t_{i+1}\|-\|t_i-\gamma_{i+1}(w_{i+1})\|> 4-2\eps-(2-\eps)=2-\eps.$$
    We thus have $\|\gamma_{i+1}(w_{i+1})-t_{i+1}\|>2-\eps>\|\gamma_{i+1}(w_{i+1})-t_i\|$, and since the prefixes of $\gamma_{i+1}$ and $\gamma'_{i+1}$ are identical, we get $|\gamma'_{i+1}|<|\gamma_{i+1}|$.
    
    We conclude that for every path $\gamma_i$, $1\le i\le\ell$, that correspond to an edge of the cycle, we can construct a path $\gamma_i'\subset \F$ from $s_{i+1}$ to $t_i$ (or from $s_1$ to $t_\ell$ when $i=\ell$) such that $|\gamma'_i|<|\gamma_i|$. This results in a valid set path of shorter total length, in contradiction to the optimality of $\Gamma$.
\end{proof}
    
\subsection{Dealing with a blocked geodesic path}\label{sec:switch_path}
Let $\Gamma=\{\gamma_1,\dots,\gamma_m\}$ be an optimal-assignment path set for $S,T,\W$, and assume w.l.o.g. that $\gamma_i(0)=s_i$ for every $1\le i\le m$.
Consider a geodesic path $\gamma_i$ from $s_i$ to $t_j$, and assume that is it $\eps$-blocked. Our goal in this section is to describe the switching process: we find another start point $s_k\in S$ and a path $\gamma'_k\subset\F$ from $s_k$ to $t_j$ such that $\gamma'_k$ is not $\eps$-blocked. Moreover, we describe a path in $\F$ from $s_i$ to $t_\ell$, where $t_\ell=\gamma_k(1)$ (see \Cref{fig:blocking_robot_switch}). The sum of lengths of these new paths will not be much longer than the original.

\begin{figure}[h!]
    \centering
    \includegraphics[page=2,scale=0.9]{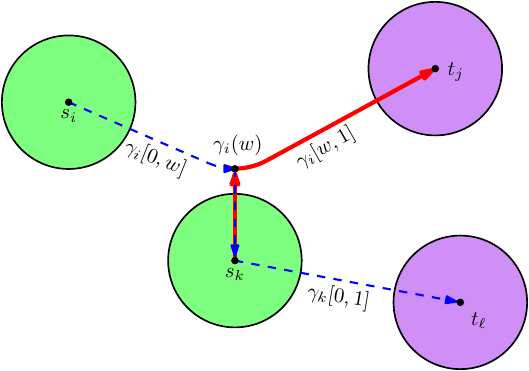}
    \caption{The start position $s_k$ is the last to $\eps$-block $\gamma_i$. The switch paths are $\gamma_i'$ (dashed blue) and $\gamma_k'$ (red).}
    \label{fig:blocking_robot_switch}
\end{figure}

For a position $s_k\in S$ that $\eps$-blocks $\gamma_i$, denote by $0\le w_{i,k}\le 1$ the largest value such that $\gamma_i(w_{i,k})$ is the point on $\gamma_i$ locally closest to $s_k$. We say that $s_k\in S$ is the \dfn{last to $\eps$-block $\gamma_i$} if for any $s_x\in S\setminus\{s_k\}$ that $\eps$-blocks $\gamma_i$, we have $w_{i,k}\ge w_{i,x}$ (i.e. the last locally closest point to $s_k$ on $\gamma_i$ is the closest to $\gamma_i(1)$ along $\gamma_i$).

Let $s_k$ be the last to $\eps$-block $\gamma_i$, and set $w=w_{i,k}$.
Consider the paths $\gamma'_k=\overline{s_k\gamma_i(w)}\circ \gamma_i[w,1]$ and $\gamma'_i=\gamma_i[0,w]\circ\overline{\gamma_i(w)s_k}\circ \gamma_k$ (where $\gamma_k$ is the geodesic path from $s_k$ to $t_\ell$, see \Cref{fig:blocking_robot_switch}). We call $\gamma'_i$ and $\gamma'_k$ the \dfn{switch paths} of $i$ and $k$, respectively.
Notice that $\|s_k- \gamma_i(w)\|<2-\eps$, and therefore $|\gamma'_k|+|\gamma'_i|<|\gamma_k|+|\gamma_i|+4-2\eps$.
The following theorem will help us in showing that the switch paths are not blocked.
\begin{restatable}{theorem}{switchPaths}\label{lem:switch_path_are_legal}
    Assume that \Cref{con:boundary_seperation_to_clear_R,con:robot_seperation_to_ensure_R,con:robot_separation_4-2eps} hold. Then 
    (i) both $\gamma'_i$ and $\gamma'_k$ are in $\F$, 
    (ii) $\overline{s_k\gamma_i(w)}$ is not $\eps$-blocked by any $p\in T\cup S\setminus\{s_k\}$,
    (iii) $\gamma'_k$ is not $\eps$-blocked by any $s_j\in S\setminus\{s_k\}$, and
    (iv) if $t_j$ is not $\eps$-blocking $\gamma_k$ then it does not $\eps$-block $\gamma'_i$.
\end{restatable}
\begin{proof}
    To show that (i) holds we only need \Cref{con:boundary_seperation_to_clear_R,con:robot_seperation_to_ensure_R} in order to apply \Cref{lem:boundary_seperation_rho}, then we have $\overline{s_k\gamma_i(w)}\subset \F$, and therefore both $\gamma'_i$ and $\gamma'_k$ are also contained in $\F$.

    For (ii) we also need \Cref{con:robot_separation_4-2eps}. Assume by contradiction that there exists $p\in T\cup S\setminus\{s_k\}$ that $\eps$-blocks $\overline{s_k\gamma_i(w)}$. Then there exists a point $x$ on $\overline{s_k\gamma_i(w)}$ such that $\|x-p\|< 2-\eps$. Because $\|s_k-x\|\le\|s_k-\gamma_i(w)\|<2-\eps$, we get by the triangle inequality that $\|s_k-p\|\le \|s_k-x\|+\|x-p\|< 4-2\eps$, in contradiction to \Cref{con:robot_separation_4-2eps}.

    To show that (iii) holds, we only need to prove that $\gamma_i[w,1]$ is not $\eps$-blocked by any position $s_j\in S\setminus\{s_k\}$. Here we also use \Cref{con:robot_separation_4-2eps}.
    Assume by contradiction that there exists $s_x\in S\setminus\{s_k\}$ that $\eps$-blocks $\gamma_i[w,1]$. Then there exists a value $w\le w'\le 1$ such that $s_x$ and $\gamma(w')$ are $\eps'$-overlapping, for some $\eps'>\eps$. Because $s_k$ is the last to $\eps$-block $\gamma_i$, we have $w_{i,x}\le w$. Since $w_{i,x}$ is the largest value such that $\gamma_i(w_{i,x})$ is the point on $\gamma_i$ locally closest to $s_x$, it must be that $s_x$ and $\gamma_i(w)$ are $\eps''$-overlapping for some $\eps''>\eps$. Therefore by the triangle inequality, $\|s_k-s_x\|\le \|s_k- \gamma_i(w)\|+\|\gamma_i(w)-s_x\|< 2-\eps+2-\eps=4-2\eps$, in contradiction to \Cref{con:robot_separation_4-2eps}.

    For (iv), we only need to show that $t_j$ does not $\eps$-block $\gamma_i[0,w]$. We again use \Cref{con:robot_separation_4-2eps}.
    Assume by contradiction that $t_j$ blocks $\gamma_i[0,w]$. Then there exists a value $0\le w'\le w$ such that $t_j$ and $\gamma(w')$ are $\eps'$-overlapping, for some $\eps'>\eps$. 
    If $\gamma(w')$ is not locally closest to $t_j$, then since by (ii) $t_j$ and $\gamma(w)$ do not $\eps$-overlap, there must be a point $0\le w''< w$ locally closest to $t_j$ such that $\|t_j-\gamma(w'')\|<2-\eps$. We can therefore apply \Cref{lem:boundary_seperation_rho} and get that $\overline{t_j\gamma(w'')}\subset\F$, in contradiction to $\gamma$ being a shortest path between $s_i$ and $t_j$.
\end{proof}

\subsection{Dealing with robots interrupting a path}\label{sec:clearance_path}
In this subsection, we discuss how to move a robot along an $\eps$-interrupted path.
We describe simple motion paths for the interrupting robots to clear the way, and we analyze the
extra distance traveled by the interrupting robots.
We note that shorter paths can be shown for the interrupting robots, however, as this does not effect the correctness of the algorithm, we chose to keep this part simple and describe a simple (yet rather short) motion path. 

As mentioned at the beginning of this section, our motion plan consists of a weakly-monotone collision free path set, that is, robots moves to targets according some ordering one by one; when one robot travels from its starting position to some target, the other robots must stay in a small disk centered at their current position (which is either a start or a target position).
To ensure that an interrupting robot is able to clear the way while remaining in a disk of radius $1+\eps$ around a start/target position, we introduce two additional constraints.
\begin{constraint}\label{con:boundary_seperation_to_clear_interruptions}
    $\omega\ge 1+\eps$
\end{constraint}
\begin{constraint}\label{con:robot_seperation_to_clear_interruptions}
    $\rho\ge2+\eps$
\end{constraint}

Consider a robot $A$ traveling in constant speed along a path $\gamma_A$ from start position $s\in S$ to some target position $t\in T$.
Let $b\in S\cup T$ be either a start or a target position, occupied by another robot $B$, which is $\eps$-interrupting $\gamma_A$.
We describe a motion plan for $B$ by constructing a \dfn{clearance path} $\Tilde\gamma_B:[0,1]
\to\reals^2$, along which the robot $B$ will travel with constant speed, as follows.
\begin{figure}[h!]
    \centering
    \includegraphics[scale=0.85]{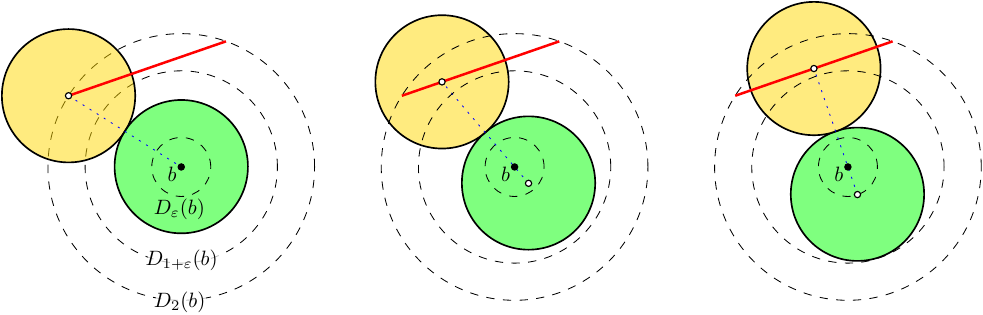}
    \caption{Left: the robot $B$ is initially positioned at $b$. In red is a maximal subpath of $\gamma_A$ where $B$ needs to clear the way. Middle, right: the robot $B$ move inside $\D_\eps(b)$ while maintaining distance exactly $2$ from the robot $A$.}
    \label{fig:polar}
\end{figure}

Let $\gamma_A[w_1,w_2]$ be a maximal subpath where $\|\gamma_A(w)-b\|\le 2$ for every $w\in[w_1,w_2]$, that is, when $A$ travels along $\gamma[w_1,w_2]$, the robot $B$ must clear the way. Note that there could be more than one such maximal subpath for $B$.
We define the corresponding subpath $\Tilde\gamma_B:[w_1,w_2]\to\reals^2$ for $B$ using a polar coordinate system, where $b$ is the pole. 
Denote by $r_x$ the radius and by $\theta_x$ the angle of a point $x$, and let $\Tilde\gamma_B(w):=\left(2-r_{\gamma_A(w)},\theta_{\gamma_A(w)}+\pi\right)$. 

The next lemma shows that the clearance path of $B$ is feasible and relatively short.
\begin{restatable}{lemma}{clearancePath}\label{clm:clearance_path_in_disk}
     $\Tilde\gamma_B[w_1,w_2]$ is a continuous path that starts and ends at $b$, such that (i) $\Tilde\gamma_B[w_1,w_2]\subset D_{\eps}(b)\subset \F$, and (ii) $|\Tilde\gamma_B[w_1,w_2]|\le|\gamma_A[w_1,w_2]|$ for $0\le \eps\le1$.
\end{restatable}
\begin{proof}
Observe that indeed $\Tilde\gamma_B[w_1,w_2]$ is a continues path because the function $f(w)=(r_{\gamma_A(w)},\theta_{\gamma_A(w)})$ for $w\in[w_1,w_2]$ is continuous. Notice that for all $w\in[w_1,w_2]$, we have $\|\gamma_A(w)-\Tilde\gamma_B(w)\|=2$, and therefore it starts at $b$ when the path $\gamma_A$ enters the disk $D_2(b)$, and ends at $b$ when the path $\gamma_A$ leaves the disk $D_2(b)$.

Because $b$ is $\eps$-interrupting $\gamma_A$, the largest distance that it has to move away from $b$ to keep $\|\gamma_A(w)-\Tilde\gamma_B(w)\|=2$ is $\eps$, and therefore we have $\Tilde\gamma_B[w_1,w_2]\subset D_{\eps}(b)$. By \Cref{con:boundary_seperation_to_clear_interruptions}, $D_{1+\eps}(b)\subset \W$ and therefore $D_{\eps}(b)\subset \F$.

Notice that throughout the motion we keep the angle between $\gamma_A(w)$ and $\Tilde\gamma_B(w)$ (relative to the pole $b$) exactly $\pi$. Thus, the angle of the paths changes at the same pace as $w$ changes. On the other hand, the sum of radii of $\gamma_A(w)$ and $\Tilde\gamma_B(w)$ throughout the motion is exactly $2$, but since $\eps\le 1$ we get that the radius of ${\Tilde\gamma_B(w)}$ is at most the radius of $\gamma_A(w)$ for all $w\in[w_1,w_2]$. Therefore $|\Tilde\gamma_B[w_1,w_2]|\le|\gamma_A[w_1,w_2]|$.
\end{proof}

Recall that $B$ may interrupt $\gamma_A$ in more than one maximal subpath, so overall we define the clearance path $\Tilde\gamma_B:[0,1]\to\reals^2$ as \[
\Tilde\gamma_B(w) = 
\begin{cases}
b, & \text{if } \|b-\gamma_A(w)\| \ge 2,\\
\left(2-r_{\gamma_A(w)},\theta_{\gamma_A(w)}+\pi\right),  & \text{if } \|b-\gamma_A(w)\| < 2.
\end{cases}
\]
In addition, we construct such a clearance path for any other robot that interrupts $\gamma_A$. 
The following theorem shows that clearance paths are collision-free, assuming \Cref{con:boundary_seperation_to_clear_R,con:robot_seperation_to_ensure_R,con:robot_separation_4-2eps,con:boundary_seperation_to_clear_interruptions,con:robot_seperation_to_clear_interruptions}.

\begin{restatable}{theorem}{clearancePathNoCollision}\label{thm:interrupted_path}
    Let $A$ be a robot traveling in constant speed along a path $\gamma_A$ from start position $s\in S$ to some target position $t\in T$, and let $\{\Tilde\gamma_B\mid B\ne A\}$ be the set of clearance paths defined for the other robots. If \Cref{con:boundary_seperation_to_clear_R,con:robot_seperation_to_ensure_R,con:robot_separation_4-2eps,con:boundary_seperation_to_clear_interruptions,con:robot_seperation_to_clear_interruptions} hold for $0\le \eps\le 1$, then the motion plan that corresponds to the paths $\gamma_A\cup \{\Tilde\gamma_B\mid B\ne A\}$ is collision-free.
\end{restatable}
\begin{proof}
    Let $B,C\ne A$ be two robots placed on $b,c\in S\cup T$, respectively. We show that the motion plan that corresponds to the paths $\Tilde\gamma_B,\Tilde\gamma_C$ and $\gamma_A$ is collision-free.
    
    By the construction, for every robot $B\ne A$ the motion plan that corresponds to the paths $\gamma_A$ and $\Tilde\gamma_B$ is collision-free.
    We now need to prove that two robots $B,C\ne A$ traveling along their respective paths $\Tilde\gamma_B$ and $\Tilde\gamma_C$ with constant speed are not colliding.

    Recall that by the construction of the paths we have $\Tilde\gamma_B\subset D_\eps\left(b\right)$ and $\Tilde\gamma_C\subset D_\eps\left(c\right)$. 
    By constraint \ref{con:robot_seperation_to_clear_interruptions}, $|b-c|\ge 2+\eps$. Therefore, for every $0\le w\le 1$ such that either $\Tilde\gamma_B(w)=b$ or $\Tilde\gamma_C(w)=c$, the robots do not collide. 
    Consider a value $0\le w\le 1$ for which both $\Tilde\gamma_B(w)\neq b$ and $\Tilde\gamma_C(w)\neq c$. Since $B$ moves along the ray from $\gamma_A$ to $b$, and $C$ moves along the ray from $\gamma_A$ to $c$, the distance between $A$ and $B$ can only grow larger at each step of time, and thus $B$ and $C$ do not collide.
\end{proof}

\section{The Algorithm}\label{sec:algorithm}
In this section we present and analyze the algorithm based on the tools developed in the previous section. The input for the algorithm is an overlap parameter $\eps$, a workspace $\W$, a set $S$ of starting positions and a set $T$ of target positions, such that \Cref{con:boundary_seperation_to_clear_R,con:robot_seperation_to_ensure_R,con:robot_separation_4-2eps,con:boundary_seperation_to_clear_interruptions,con:robot_seperation_to_clear_interruptions} hold. In addition, we assume that  each connected component of $\F$ contains an equal number of start and target positions (as otherwise, there is no solution).

For the simplicity of presentation and analysis, in this section we provide an iterative version of the algorithm. In the $i$'th iteration, for $1\le i\le m$, the algorithm constructs a collision-free motion plan for all the robots over the time interval $[i-1,i]$, in which a single robot is moving from a start position to a target position along an $\eps$-interrupted path (\Cref{sec:switch_path}), and the other robots may move in their small neighborhood to clear the way (\Cref{sec:clearance_path}). The output path for a robot $A$ will then be the concatenation $\gamma^*_A:[0,m]\to\F$ of the paths $\gamma_A:[i-1,i]\to\F$ constructed for it in all iterations. We then update $S$ and $T$ according to the path that was chosen, and remove from $\W$ a disk of radius $1-\eps$ around the target.

At the beginning of this section we gave a high level description of our algorithm, and a more detailed pseudo code is provided in \Cref{alg:eps_overlap_algo}. Below, we prove that this algorithm returns a collision-free motion plan, with total length $O(\OPT)$, where  $\OPT$ is the optimal size of a solution to the problem.

\begin{algorithm}[h]
\caption{Motion Plan for Unlabeled Disks}\label{alg:eps_overlap_algo}
\KwIn{$\W$, $S$, $T$, and $0\le \eps\le 1$ such that \Cref{con:boundary_seperation_to_clear_R,con:robot_seperation_to_ensure_R,con:robot_separation_4-2eps,con:boundary_seperation_to_clear_interruptions,con:robot_seperation_to_clear_interruptions} hold}
\KwOut{A collision-free motion plan}
\begin{itemize}
    \item For each robot $A\in [m]$, initialize an empty path $\gamma_A^*$.
    \item While $S\ne \emptyset$:
    \begin{enumerate}
        \item Compute an optimal-assignment path set $\Gamma$ for $S,T,\W$.  
        \item Find an $\eps$-interrupting target $t_j$, and let $\gamma_i\in \Gamma$ be the path from $s_i$ to $t_j$.
        \item If $\gamma_i$ is not $\eps$-blocked, set $A \gets i$, $s_A\gets s_i$, $\gamma_A \gets \gamma_i$.
        \item Else, let $s_k$ be the last to $\eps$-block $\gamma_i$, and set $A \gets k$, $s_A\gets s_k$, and $\gamma_A\gets \gamma'_k$, \\where $\gamma'_k$ is the switch path (\Cref{sec:switch_path}).
        \item For all $B \neq A$, set $\gamma_B\gets \Tilde{\gamma}_B$, where $\Tilde{\gamma}_B$ is the clearance path (\Cref{sec:clearance_path}). 
        \item For each robot $B\in [m]$ set $\gamma^*_B\gets \gamma^*_B\circ \gamma_B$.
        \item Set $\W\gets\W\setminus\D_{1-\eps}(t_j)$, $S\gets S\setminus\{s_A\}$, $T\gets T\setminus\{t_j\}$.
    \end{enumerate}
    \item Return the set of paths $(\gamma_1^*,\dots,\gamma_m^*)$.
\end{itemize}
\end{algorithm}

\begin{theorem}
    Consider an input $\W$, $S$, $T$, and $0\le \eps\le 1$ such that \Cref{con:boundary_seperation_to_clear_R,con:robot_seperation_to_ensure_R,con:robot_separation_4-2eps,con:boundary_seperation_to_clear_interruptions,con:robot_seperation_to_clear_interruptions} hold. If each connected component of $\F$ contains an equal number of start and target positions, then \Cref{alg:eps_overlap_algo} is guaranteed to return a collision-free motion plan $\Gamma^*$.
\end{theorem}
\begin{proof}
    First, observe that the solution returned by the algorithm is a path set that matches every start position in $S$ to a unique target in $T$, because in each iteration we find a path between some $s\in S$ and $t\in T$ and then remove them from $S$ and $T$ for the next iteration.

    Consider the first iteration of the algorithm. 
    By \Cref{theorm:almost-standalone}, an $\eps$-interrupting target $t_j$ exists, and let $\gamma_i$ be the path from $s_i$ to $t_j$.
    If $\gamma_i$ is $\eps$-blocked, let $s_k\in S$ be the last to block $\gamma_i$, and by \Cref{lem:switch_path_are_legal}(iii), the switch path $\gamma'_k$ is in $\F$, and it is not blocked by any other robot. Otherwise, $\gamma_i$ is not blocked.
    In any case we choose an $\eps$-interrupted path that leads to $t_j$, and by \Cref{thm:interrupted_path} the motion plan computed in this iteration is collision-free.

    In the following iteration, after removing $D_{1-\eps}(t_j)$ from $\W$, \Cref{con:boundary_seperation_to_clear_R,con:robot_seperation_to_ensure_R,con:robot_separation_4-2eps,con:boundary_seperation_to_clear_interruptions,con:robot_seperation_to_clear_interruptions} still hold for the updated sets $S$ and $T$: the robots-separation does not change, so
    \Cref{con:robot_seperation_to_ensure_R,con:robot_separation_4-2eps,con:robot_seperation_to_clear_interruptions} still hold, and specifically by \Cref{con:robot_separation_4-2eps}
    for every $p_1,p_2\in S\cup T$ we have $\|p_1-p_2\|\ge 4-2\eps$. Therefore, for every $p\in S\cup T\setminus\{t_j\}$ and $x\in D_{1-\eps}(t_j)$ we have $\|p-x\|\ge 3-\eps$, which satisfies both \Cref{con:boundary_seperation_to_clear_R} ($\|p-x\|\ge \sqrt{(3-\sqrt{3}-\eps)^2+1}$) and \Cref{con:boundary_seperation_to_clear_interruptions} ($\|p-x\|\ge 1+\eps$), for any $0\le \eps\le 1$. 
    
    However, we also need to take into account the robot that is now positioned on $t_j$. 
    The main observation that we will use is that $t_j$ does not block any path $\gamma$ in the free space of $\W\setminus D_{1-\eps}(t_j)$; the reason is that the distance between any point on $a\in \gamma$ and a point $p\in D_{1-\eps}(t_j)$ is at least $1$, and therefore the $\|t_j-a\|\ge 2-\eps$ which means that the overlap between $t_j$ and $a$ is at most $\eps$.

    Consider an iteration $a>1$ of the algorithm, and denote by $\W_a,S_a,T_a$ the sets $\W,S,T$ at the beginning of iteration $a$.
    The first step is to compute an optimal-assignment path set $\Gamma_a$.
    Let $\Gamma_{a-1}$ be the optimal-assignment path set for $\W_{a-1},S_{a-1},T_{a-1}$.
    At the beginning of iteration $a$, we have $\W_a=\W_{a-1}\setminus\D_{1-\eps}(t_j)$ and $T_a=T_{a-1}\setminus\{t_j\}$, where $t_j$ is an $\eps$-interrupting target in $\Gamma_{a-1}$. Since $t_j$ does not block any path in $\Gamma_{a-1}$, the disk $\D_{1-\eps}(t_j)$ does not overlap with the trace of any path in $\Gamma_{a-1}\setminus\{\gamma_i\}$. Therefore, any path $\gamma\in\Gamma_a\setminus\{\gamma_i\}$ is in the free space of $\W_a$.
    If $\gamma_i$ was not blocked in iteration $a-1$, then $S_a=S_{a-1}\setminus\{s_i\}$, and thus $\Gamma_a\setminus\{\gamma_i\}$ corresponds to a matching of $S_a$ and $T_a$, which implies that $\Gamma_a$ exists. Otherwise, $S_a=S_{a-1}\setminus\{s_k\}$, where $s_k\in S_{a-1}$ is the last to $\eps$-block $\gamma_i$. Let $t_\ell$ be the target that was matched to it in $\Gamma_{a-1}$, and let $\gamma_k\in \Gamma_{a-1}$ be the path from $s_k$ to $t_\ell$. 
    Observe that $\Gamma_{a-1}\setminus \{\gamma_i,\gamma_k\}$ corresponds to a matching between $S_a\setminus\{s_i,s_k\}$ and $T_{a-1}\setminus\{t_k,t_\ell\}$, and thus to have a matching between $S_a$ and $T_a$ we need to find a path in the free space of $\W_a$ between $s_i$ and $t_\ell$. We show that the switch path $\gamma'_i$ from $s_i$ to $t_\ell$ as defined in \Cref{sec:switch_path}, has this property. Recall that $\gamma_i'$ is a concatenation of $\gamma[1,w]$, $\overline{\gamma_i(w)s_k}$, and $\gamma_k$. The paths $\gamma_i,\gamma_k$ are in $\Gamma_{a-1}$ and therefore lie in the free space of $\W_{a-1}$. By  \Cref{lem:switch_path_are_legal}(i), the entire path $\gamma_i'$ is in the free space of $\W_{a-1}$. We thus need to show that the disk $\D_{1-\eps}(t_j)$ does not overlap with the trace of $\gamma_i'$. Indeed, since $t_j$ is an $\eps$-interrupting target, it does not block $\gamma_k$, and therefore by \Cref{lem:switch_path_are_legal}(iv) it is also not blocking $\gamma_i'$.

    The next step is to find an $\eps$-interrupting target $t_j$. For this, notice that \Cref{theorm:almost-standalone} still holds, because it considers only start and target positions in $S_a$ and $T_a$.
    Let $\A$ be the set of robots that are already lying on target positions.
    Notice that we can also apply \Cref{lem:switch_path_are_legal} on $\W_a,S_a,T_a$ and $\Gamma_a$, and we just need to make sure that no robot in $\A$ can block the switch path $\gamma'_k$. As mentioned above, this is true in general, that is, for any path $\gamma\in \Gamma_a$ from $s\in S_a$ to $t\in T_a$, no robot in $\A$ can block $\gamma$, because the overlap between any such robot and any point on $\gamma$ is at most $\eps$. In other words, robots in $\A$ can only interrupt paths in $\Gamma_a$.

     The final step is using \Cref{thm:interrupted_path} to construct the clearance paths. Notice that this construction applies for any $p\in S\cup T$, and therefore it can be applied for the $\eps$-interrupted path, even if there were robots in all the start and target positions, except for the endpoints of the path.

     Lastly, notice that the concatenation of the paths that correspond to a robot $A$ from all the iterations of the algorithm results in a continuous path. This is because each clearance path starts and ends at the same start/target position.
\end{proof}

\begin{restatable}{theorem}{algoLength}\label{thm:algo_length_and_runtime}
    \Cref{alg:eps_overlap_algo} runs in $O(m^4+m^2n^2)$ time, and returns a solution $\Gamma^*$ such that $|\Gamma^*|=O(\OPT)$, where $\OPT$ is the size of an optimal solution.
\end{restatable}
\begin{proof}    
    To bound the sum of distances we separate it to two parts - (1) the distances traveled by robots while traveling from their starting position to a target position and (2) the distances of clearance paths.

    Starting with (1), let $\Gamma_a$ be the optimal assignment that the algorithm found in the $a$'th iteration and let $\Gamma^*_a$ be the paths assigned in the iterations before the $a$'th iteration (not include the clearance paths). We claim that $|\Gamma_{a+1}\cup\Gamma^*_{a+1}|\le|\Gamma_a\cup\Gamma^*_a|+4-2\eps$. To see that, let  $t_j$ be the $\eps$-interruption target found in the $a$'th iteration and let $s_i$ be the starting position such that the path $\gamma_i$ leads to $t_j$. If this path is $\eps$-interrupted then $\Gamma_{a+1}^*=\Gamma_a^*\cup\{\gamma_i\}$ and since $\Gamma_a\setminus\{\gamma_i\}$ is feasible assignment path set in the next iteration then $|\Gamma_{a+1}|\le|\Gamma_a\setminus\{\gamma_i\}|$ and indeed, $|\Gamma_{a+1}\cup\Gamma^*_{a+1}|\le|\Gamma_a\cup\Gamma^*_a|\le|\Gamma_a\cup\Gamma^*_a|+4-2\eps$.
    
    Otherwise $\gamma_i$ is $\eps$-blocked. Let $s_k$ be the last $\eps$-blocking target. Recall the switch paths defined  in \Cref{sec:switch_path}. In this case, $\Gamma_{a+1}^*=\Gamma_a^*\cup\{\gamma_k'\}$ and since $(\Gamma_a\setminus\{\gamma_i,\gamma_k\})\cup\{\gamma_i'\}$ is feasible assignment path set in the next iteration then $|\Gamma_{a+1}|\le|(\Gamma_a\setminus\{\gamma_i,\gamma_k\})\cup\{\gamma_i'\}|$ and we get $|\Gamma_{a+1}\cup\Gamma^*_{a+1}|\le|((\Gamma_a\setminus\{\gamma_i,\gamma_k\})\cup\{\gamma_i'\})\cup\Gamma_a^*\cup\{\gamma_k'\}|$. By the definition of the switch paths we get $|\gamma_i|+|\gamma_k|=|\gamma_i'|+2-\eps+|\gamma_k'|+2-\eps$ and we conclude that $|\Gamma_{a+1}\cup\Gamma^*_{a+1}|\le|\Gamma_a\cup\Gamma_a^*|+4-2\eps$.

    Applying the inequality $m$ times we get that $|\Gamma_m\cup\Gamma_m^*|\le|\Gamma_1\cup\Gamma_1^*|+(4-2\eps)m$ and since $\Gamma_1^*,\Gamma_m=\emptyset$ and since $|\Gamma_1|\le\OPT$ (as $\Gamma_1$ is the best assignment path set even when ignoring collision and $\OPT$ is a size of a assignment path set) we conclude that $|\Gamma_m^*|\le\OPT+(4-2\eps)m$.

    As for the extra distance traveled by the clearance paths, let $\gamma_A$ be the path that leads to the $\eps$-interrupting target at some iteration. Let $p$ be starting/target position occupied by a robot $B$ that interrupt $\gamma_A$. For a maximal subpath $\gamma_A[w_1,w_2]$ where $\|\gamma_A(w)-p\|\le2$ for every $w\in[w_1,w_2]$, the robot on $p$ traveled according to $\tilde\gamma_B[w_1,w_2]$. \Cref{clm:clearance_path_in_disk} shows that $\tilde\gamma_B[w_1,w_2]\le\gamma_A[w_1,w_2]$. We mark with $c$ the maximal number of robots that can interrupt $\gamma_A$ at a single location, that without any separation assumptions its $6$, and conclude that $\sum_{B\ne A}|\Tilde\gamma_B|\le c|\gamma_A|$. So by adding the extra distance traveled in the clearance paths in all the iterations we get that $|\Gamma^*|\le(1+c)\Gamma_m^*\le(1+c)(\OPT+(4-2\eps)m)=O(\OPT)$. Note that this analysis is not tight; see \Cref{rem:solution_size} for more details.
    
    Regarding the running time, note that except for computing the clearance paths in each iteration, the rest of the operations are changed version of operations of the algorithm in \cite{SYZH15} and the implementation analysis in \cite{SYZH15} provide valid analysis to our operations, yields $O\left(m^4+m^2n^2\right)$ for the changed operations. As for computing the clearance paths, assume in specific iteration that we have already the path $\gamma_A$ that lead a robot to an $\eps$-interruption target and we need to compute only the clearance paths. The path consist of segments and arcs and its complexity is bound by the complexity of $\F$ which is $O\left(m+n\right)$ (as in each iteration we add exactly one disk to $\O$). For each segment or arc, computing the clearance path corresponding to it for a single robot takes $O\left(1\right)$ time, and in total  $O\left(m+n\right)$ time. So for $m$ robots in $m$ iterations, the total running time for computing the clearance paths is $O\left(m^3+m^2n\right)$ which is absorbed in the $O\left(m^4+m^2n^2\right)$.
\end{proof}

We conclude with the following theorem.
\begin{theorem}\label{thm:opt_sep_bounds}
    Given a polygonal workspace $\W$ with $n$ edges, a set $S$ of $m$ starting positions and a set $T$ of $m$ target positions, there is an algorithm that computes in $\Tilde{O}(m^4+n^2m^2)$ time a collision-free motion plan of length $O(\OPT)$, where $\OPT$ is the length of an optimal solution, or reports that such a motion plan does not exist, assuming that $\rho=\max\{4-2\eps,2+\eps\}$ and $\omega=\max\{1+\eps,\sqrt{(3-\sqrt{3}-\eps)^2+1}\}$.
\end{theorem}

\subsection{Optimizing the separation bounds}\label{sec:optimzing_bounds}
We now optimize the overlap parameter over \Cref{con:boundary_seperation_to_clear_R,con:robot_seperation_to_ensure_R,con:robot_separation_4-2eps,con:boundary_seperation_to_clear_interruptions,con:robot_seperation_to_clear_interruptions}, to achieve either the smallest obstacles-separation or the smallest robots-separation possible.
Recall that \Cref{con:boundary_seperation_to_clear_R,con:boundary_seperation_to_clear_interruptions} specify the required obstacles-separation bounds, while \Cref{con:robot_seperation_to_ensure_R,con:robot_separation_4-2eps,con:robot_seperation_to_clear_interruptions} specify the required robots-separation bounds (see \Cref{fig:constraints}).

We highlight three interesting values of $\eps$:
\begin{itemize}
    \item For $\eps=0$ the motion plan is monotone,  
    and the length of the solution is at most $\OPT+4m$.
    In this case, we obtain $\omega=\sqrt{13-6\sqrt{3}}\approx1.614$ and $\rho=4$.
    \item For $\eps>0$, the motion plan is weakly-monotone.
    \item For $\eps=\frac{15-6\sqrt{3}}{13}\approx0.354$, the obstacles-separation is minimized, and we obtain $\omega=1+\frac{15-6\sqrt{3}}{13}\approx1.354$ and $\rho=4-2\frac{15-6\sqrt{3}}{13}\approx3.291$.
    \item For $\eps=\frac{2}{3}$, the robots-separation is minimized, and we obtain $\omega=1\frac{2}{3}$ and $\rho=2\frac{2}{3}$.
\end{itemize}

\begin{figure}[h]
    \centering
    \includegraphics[scale=0.5]{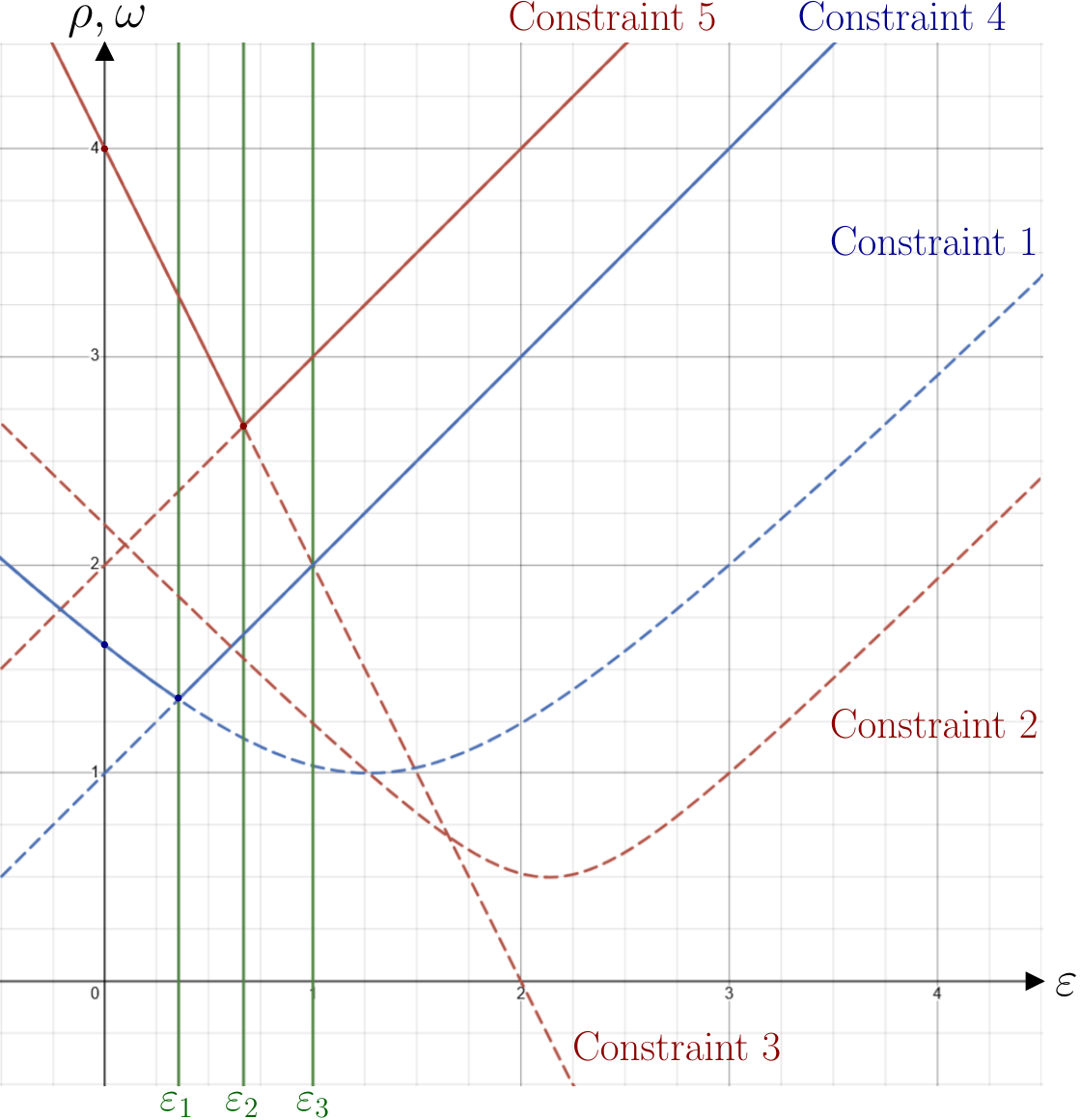}
    \caption{The constraints: in red are the constraint on $\rho$, and in blue the constraints on $\omega$. Three interesting values of are marked $\eps$: $\eps_1=\frac{15-6\sqrt{3}}{13}$,$\eps_2=\frac23$, and $\eps_3=1$ (see later in section \Cref{sec:improved}).}
    \label{fig:constraints}
\end{figure}

\subsection{A remark on the length of our solution}\label{rem:solution_size}
    Note that in fact, we can bound the size of the solution with the term $(1+c)(\OPT+(4-2\eps)m)$, where $c$ is the maximum number of start/target positions that may intersect a given unit disk under our robots-separation constraints. 
    The term $\OPT+(4-2\eps)m$ accounts for the motion of robots from their starting positions to their targets, and the $(1+c)$ factor comes from the clearance paths.

    For example, when the robots-separation is $\rho=2$ then at most $c=6$ start/target positions may intersect a given unit disk.
    When $\eps=0$ we have $\rho=4$, so no two start/target positions intersect a single unit disk, and therefore the size of the solution is $OPT+4m$ --- near optimal solution as in \cite{SYZH15}.
    
    We also note that this bound can be significantly reduced with a tighter analysis for two reasons: (1) even if a robot can collide with $c$ robots at some point, along most of its path it collides with fewer than $c$; and (2) our analysis of the additional distance traveled in the clearance paths assumes $\eps\le1$, but if $\eps\ll1$, the extra distance traveled during the clearance paths is much smaller.

\section{(Almost) Optimal Robot Separation}\label{sec:improved}
In this section we further develop the algorithm from the previous section by adding a new parameter, so that the robot separation can be reduced down to (almost\footnote{Recall that in our setting we have $\rho\ge 2$, however, by distinguishing bichromatic and monochromatic separation as in~\cite{BBBBFHKOS22}, one may be able to achieve better robot separation bounds.}) the minimal value, i.e., $\rho=2$. This would come at the cost of increasing the length of the solution and $\omega$.

Recall that \Cref{con:robot_seperation_to_ensure_R,con:robot_separation_4-2eps,con:robot_seperation_to_clear_interruptions} are the constraints related to $\rho$.
\Cref{con:robot_seperation_to_ensure_R} is needed in the proof of \Cref{lem:boundary_seperation_rho}, which is used in order to show that an $\eps$-blocking robot can move into the trace of the selected path and make a switch. Notice that this constraint is dominated by the maximum of \Cref{con:robot_separation_4-2eps,con:robot_seperation_to_clear_interruptions} for every $\eps\le 1$ (see \Cref{fig:constraints}).
\Cref{con:robot_separation_4-2eps} is needed for two reasons. The first is for the triangle inequality in the proof of \Cref{theorm:almost-standalone}, to show that an $\eps$-interrupting target always exists. The second is in \Cref{lem:switch_path_are_legal}, to show that the switch path is not $\eps$-blocked, and that the selected target is not $\eps$-blocking it.
\Cref{con:robot_seperation_to_clear_interruptions} is needed for the clearance paths to exist, i.e., in order for the other robots to move away from the selected $\eps$-interrupted path.

Notice that in \Cref{sec:clearance_path}, where we describe the clearance paths, only robots that $\eps$-interrupt the selected path are clearing the way, and the amount of ``separation'' that they need for this task is exactly what \Cref{con:robot_seperation_to_clear_interruptions} captures. However, we may allow other robots (that do not directly interrupt the path) to move away as well, therefore enable to clear the path with less separation required for each robot. In other words, the amount of separation is ``split'' between the robots that move in order to clear the path.

To this end, we introduce a new parameter $0\le \alpha\le 1$, which will be used in order to balance the amount of separation that we ``steal'' from the other robots in order to clear the selected path. We then replace  \Cref{con:robot_seperation_to_clear_interruptions} by
\begin{constraint}\label{con:robot_seperation_to_clear_interruptions_improved}
    $\rho\ge2+\alpha\eps$
\end{constraint}

For $\alpha=1$, \Cref{con:robot_seperation_to_clear_interruptions_improved} is equal to \Cref{con:robot_seperation_to_clear_interruptions}, and for $\alpha=0$, we get that \Cref{con:robot_seperation_to_clear_interruptions_improved} is $\rho\ge2$. Notice that \Cref{con:robot_seperation_to_ensure_R} is still dominated by the maximum of \Cref{con:robot_separation_4-2eps,con:robot_seperation_to_clear_interruptions_improved} for any value of $0\le\alpha\le1$ (and $\eps\le 1$).

The only change in our algorithm will be in the way we are dealing with $\eps$-interrupted paths, i.e., we replace the construction of clearance paths that was described in \Cref{sec:clearance_path} by the construction below.

\subsection{Clearing an interrupted path with less separation} \label{sec:clearance_path2}
Consider a robot $A$ traveling along an $\eps$-interrupted path $\gamma_A$ from start position $s\in S$ to some target position $t\in T$.
Assuming \Cref{con:boundary_seperation_to_clear_R,con:robot_seperation_to_ensure_R,con:robot_separation_4-2eps,con:boundary_seperation_to_clear_interruptions,con:robot_seperation_to_clear_interruptions_improved}, we describe simple collision-free motion paths for (some of) the robots to clear the way for the robot $A$, and we analyze the extra distance traveled by them.

Note that in this modified version, our motion plan may not be weakly-monotone: a robot that $\eps$-overlaps $\gamma_A$ has to move a distance of $\eps$ in order to clear the path, however, for $\alpha<1$, the disk of radius $2+\eps$ around it is not guaranteed to be empty of other robots.

Let $\gamma_A(w)$ be the location of $A$ on $\gamma_A$ at time $w$, then $\gamma_A(0)=s$ and $\gamma_A(1)=t$.
Consider the overlap at time $w$ between the current position of $A$ and the closest start/target position, 
denoted 
$$
\delta(w)=\max\left\{0,2-\min_{p\in S\cup T}\|\gamma_A(w)-p\|\right\}.
$$

Consider another robot $B$ positioned at some start/target position $b\in S\cup T$. We describe a motion plan for $B$ by constructing a \dfn{clearance path} $\Tilde\gamma_B:[0,1]
\to\reals^2$, along which the robot $B$ will travel, as follows.

We first define the function $f_b(w)$ to return the amount of movement needed for the robot $B$ to clear the way for $A$, depending on the distance from its original position $b$ to the position of $A$ at time $w$, and the maximal overlap at time $w$, as follows:
$$f_b(w)= \frac{\alpha\eps}{2+\alpha\eps}(2-\delta(w)-\|\gamma_A(w)-b\|)+\delta(w).$$
Notice that $\|\gamma_A(w)-b\|\ge 2-\delta(w)$, because $2-\delta(w)$ is the distance to the closest start/target position, and thus the first addend in $f_b(w)$ is non-positive. Therefore, if $\delta(w)=0$ (i.e., there is no start/target position overlapping with $A$ at time $w$), then $f_b(w)\le0$. 
If $\delta(w)>0$, then $f_b(w)>0$ whenever $\|\gamma_A(w)-b\|<2+\frac{2}{\alpha\eps}\delta(w)$.

We then define the \dfn{clearance path} $\Tilde\gamma_B:\to\reals^2$ for $B$ using a polar coordinate system, where $b$ is the pole. 
Denote by $r(w)$ the radius and by $\theta(w)$ the angle of the point $\gamma_A(w)$ w.r.t. the pole $b$. That is, $r(w)=\|\gamma_A(w)-b\|$, is the distance between the position of $A$ at time $w$, and $b$ (the start/target position on which $B$ was initially standing).
Let 
$$\Tilde\gamma_B(w):=\left(\max\{f_b
(w),0\},\theta(w)+\pi\right).$$

Let $\gamma_A(w_1,w_2)$ be a maximal subpath where  $f_b(w)>0$ for every $w\in[w_1,w_2]$, that is, when $A$ travels along $\gamma[w_1,w_2]$, the robot $B$ will potentially move. Note that there could be more than one such maximal subpath for $B$.
The next lemma shows that the clearance path of $B$ is feasible and relatively short.
\begin{restatable}{lemma}{clearancePath}\label{clm:clearance_path_in_disk_improved}
     $\Tilde\gamma_B[w_1,w_2]$ 
     is a continuous path that starts and ends at $b$, such that (i) $\Tilde\gamma_B[w_1,w_2]\subset D_{\eps}(b)\subset \F$, and (ii) $|\Tilde\gamma_B[w_1,w_2]|\le \sqrt{2}|\gamma_A[w_1,w_2]|$ for $0\le \eps\le1$.
\end{restatable}

\begin{proof}
Observe that indeed $\Tilde\gamma_B[w_1,w_2]$ is a continuous path because the functions $f_b(w)$ and $\theta(w)$ are continuous. 
For $w\in\{0,1\}$ we have $\delta(w)=0$ (there is no overlap when $A$ is on $s$ or $t$), and therefore $f_b(w)\le 0$ and $\Tilde\gamma_B(0)=\Tilde\gamma_B(1)=b$.
Now, since $f_b(w)$ is continuous, 
and $[w_1,w_2]\subseteq[0,1]$ is a maximal range in which $f_b(w)\ge 0$, it must be that $f_b(w_1),f_b(w_2)=0$ and hence $\Tilde\gamma_B(w_1)=\Tilde\gamma_B(w_2)=b$, so $\Tilde{\gamma}_B(w)[w_1,w_2]$ starts and ends at $b$.

To prove (i), observe that since the first addend of $f_b(w)$ is non-positive, we have $f_b(w)\le\delta(w)$.
Because $\gamma_A$ is $\eps$-interrupted, the overlap $\delta(w)$ at a given moment is at most $\eps$. We conclude that $\max\{f_b(w),0\}\le \eps$ and therefore the robot $B$ remains in a disk of radius $\eps$ centered at $b$, i.e., $\Tilde\gamma_B[w_1,w_2]\subset D_{\eps}(b)$. By \Cref{con:boundary_seperation_to_clear_interruptions}, we have $D_{1+\eps}(b)\subset \W$ and hence $D_{\eps}(b)\subset \F$.

To prove (ii), first recall that $\gamma_A(w)$ is a concatenation of segments and arcs, so it is differentiable except for finitely many points. For simplicity, we assume that $\gamma_A(w)[w_1,w_2]$ is smooth, as otherwise we can partition $[w_1,w_2]$ to ranges for which $\gamma_A(w)$ is smooth and the analysis remains the same on each part. 

By definition, we have $\gamma_A(w)=(r(w),\theta(w))$, and the length of $\gamma_A[w_1,w_2]$ is 
$\int_{w_1}^{w_2}\|\gamma_A'(w)\|dw$ were $\|\gamma_A'(w)\|=\sqrt{r'(w)^2+r(w)^2\theta'(w)^2}$.
Similarly, since $f_b(w)\ge0$ for $w\in[w_1,w_2]$, we have $\Tilde\gamma_B(w)=(f_b(w),\theta(w)+\pi)$ for $w\in[w_1,w_2]$, and the length of $\Tilde\gamma_B[w_1,w_2]$ is
$\int_{w_1}^{w_2}\|\Tilde\gamma_B'(w)\|dw=\int_{w_1}^{w_2}\sqrt{f_b'(w)^2+f_b(w)^2\theta'(w)^2}dw$.

We first compare the angular components. Recall that 
$f_b(w)\le\delta(w)\le\eps$, and that $r(w)\ge 2-\eps\ge\eps$, and hence $f_b(w)\le r(w)$.
Therefore,
$f_b(w)^2\theta'(w)^2 \le r(w)^2\theta'(w)^2\le\|\gamma_A'(w)\|^2$.

Next, we compare the radial components. Denote $c=\frac{\alpha\eps}{2+\alpha\eps}$, then  
\[f_b(w)= c(2-\delta(w)-r(w))+\delta(w),\]
and therefore $|f_b'(w)|=|(1-c)\delta'(w)-c r'(w)|$.
Since $0<c<1$, we get 
\[|f_b'(w)| \le(1-c)|\delta'(w)|+c| r'(w)|.\] 
Now, since $\delta'(w)$ and $r'(w)$ each describe the change in the distance between some point and the robot $A$ at time $w$, we have that $|\delta'(w)|,|r'(w)|\le\|\gamma_A'(w)\|$, and thus $f_b'(w)^2\le \|\gamma_A'(w)\|^2$.

Combining the two inequalities, we get
$\|\Tilde\gamma_B'(w)\|^2
=
f_b'(w)^2+f_b(w)^2\theta'(w)^2
\leq
2\|\gamma_A'(w)\|^2$.
Thus, for every $w\in[w_1,w_2]$ we have
$\|\Tilde\gamma_B'(w)\| \le \sqrt{2}\|\gamma_A'(w)\|$.

Integrating over $[w_1,w_2]$ yields
\[
|\Tilde\gamma_B[w_1,w_2]|
=
\int_{w_1}^{w_2}\|\widetilde\gamma_B'(w)\|dw
\leq
\int_{w_1}^{w_2}\sqrt{2}\|\gamma_A'(w)\|dw
=
\sqrt{2}|\gamma_A[w_1,w_2]|.
\]
\end{proof}

We construct such a clearance path for any other robot that interrupts $\gamma_A$. 
\Cref{thm:interrupted_path_improved} shows that clearance paths are collision-free, assuming \Cref{con:boundary_seperation_to_clear_R,con:robot_seperation_to_ensure_R,con:robot_separation_4-2eps,con:boundary_seperation_to_clear_interruptions} and \Cref{con:robot_seperation_to_clear_interruptions_improved}. We begin with the following auxiliary lemma.

\begin{restatable}{theorem}{clearancePathNoCollision}\label{thm:interrupted_path_improved}
    Let $A$ be a robot traveling along a path $\gamma_A$ from start position $s\in S$ to some target position $t\in T$, and let $\{\Tilde\gamma_B\mid B\ne A\}$ be the set of clearance paths defined for the other robots. If \Cref{con:boundary_seperation_to_clear_R,con:robot_seperation_to_ensure_R,con:robot_separation_4-2eps,con:boundary_seperation_to_clear_interruptions,con:robot_seperation_to_clear_interruptions_improved} hold for $0\le \eps\le 1$, then the motion plan that corresponds to the paths $\gamma_A\cup \{\Tilde\gamma_B\mid B\ne A\}$ is collision-free.
\end{restatable}
\begin{proof}
    Let $B\ne A$ be a robot placed on $b\in S\cup T$. First, we show that the motion plan that corresponds to the paths $\Tilde\gamma_B$ and $\gamma_A$ is collision-free. 
    Observe that by the construction of $\tilde\gamma_B$, for every $w\in[0,1]$ the point $b$ lies on the line segment connecting $\tilde\gamma_B(w)$ and $\gamma_A(w)$. Therefore, for every $w\in[0,1]$ it holds that 
    \begin{align*}
        \|\gamma_A(w)-\tilde\gamma_B(w)\|&=\|\gamma_A(w)-b\|+\|\tilde\gamma_B(w)-b\|\\
        &=r(w)+f_b(w)=r(w)+\frac{\alpha\eps}{2+\alpha\eps}(2-\delta(w)-r(w))+\delta(w)\\
        &=2-(1-\frac{\alpha\eps}{2+\alpha\eps})(2-\delta(w)-r(w)).
    \end{align*}
    Recall that $r(w)\ge 2-\delta(w)$, and because $\eps,\alpha\in[0,1]$, we have $\frac{\alpha\eps}{2+\alpha\eps}\in[0,1]$. Therefore $(1-\frac{\alpha\eps}{2+\alpha\eps})(2-\delta(w)-r(w))\le 0$, and we have $\|\tilde\gamma_B(w)-\gamma_A(w)\|\ge2$. 

     Next, we show that a robot $C\notin\{A,B\}$ traveling along its clearance paths $\Tilde\gamma_C$, does not collide with $B$. In other words, we show that for every $w\in[0,1]$, we have $\|\Tilde\gamma_B(w)-\Tilde\gamma_C(w)\|\ge 2$.
     Denote by $c\in S\cup T$ the initial location of the robot $C$, and recall that $\Tilde\gamma_B$ and $\Tilde\gamma_C$ were defined using a polar coordinate systems centered at $b$ and $c$, respectively. For the next part of the proof, we move to a global coordinate system, as follows. 
     Fix some $w\in[0,1]$, and consider the polar coordinate system where $\gamma_A(w)$ is the pole. Then, for a point $p$ in the Cartesian coordinate system, $r(p)$ and $\theta(p)$ are the radial and angular components relative to $\gamma_A(w)$. Let 
        $$f(x)= \frac{\alpha\eps}{2+\alpha\eps}(2-\delta(w)-x)+\delta(w),$$
     then in this coordinate system, we have 
     $\Tilde\gamma_B(w)=(r(b)+\max\{0,f(r(b))\},\theta(b))$ and 
     $\Tilde\gamma_C(w)=(r(c)+\max\{0,f(r(c))\},\theta(c))$.          
     
     Assume w.l.o.g. that $f(r(b))\ge f(r(c))$.
     There are three cases:
     \begin{enumerate}
         \item 
            If $f(r(b))<0$ and $f(r(c))<0$, then $\Tilde\gamma_B(w)=(r(b),\theta(b))$ and $\Tilde\gamma_C(w)=(r(c),\theta(c))$, and we are done because $\|b-c\|\ge2$.
         \item 
             Else, if $f(r(b))\ge 0$ and $f(r(c))\ge0$. Recall that by \Cref{con:robot_seperation_to_clear_interruptions_improved} we have $\|b-c\|\ge 2+\alpha\eps$.
             Consider the transformation $T((x,\theta))=(x+f(x),\theta)$, then we need to show that $\|T((r(b),\theta(b)))-T((r(c),\theta(c)))\|\ge 2$.
             Denote $k=-\frac{\alpha\eps}{2+\alpha\eps}$, then $f(x)= kx-k(2-\delta(w))+\delta(w)$.
             Notice that $k\ge -1$, and $K=-k(2-\delta(w))+\delta(w)\ge0$.
             Therefore, by \Cref{lem:cosin} below we have $$\|T((r(b),\theta(b)))-T(r(c),\theta(c)))\|\ge (1-\dfrac{\alpha\eps}{2+\alpha\eps})\|b-c\|\ge\dfrac{2}{2+\alpha\eps}(2+\alpha\eps)=2$$ 
        \item
            Else, we have $f(r(b))\ge 0$ and $f(r(c))<0$. Note that in that case $$r(b)+f(r(b))\le2+\dfrac{2\delta(w)}{\alpha\eps}<r(c)+f(r(c))<r(c)$$
            Hence, by \Cref{fact:r1r2r3} below we get
            \begin{align*}
                \|\Tilde\gamma_B(w)-\Tilde\gamma_C(w)\|&=\|(r(b)+f(r(b)),\theta(b))-(r(c),\theta(c))\|\\
                &\ge\|(r(b)+f(r(b)),\theta(b))-(r(c)+f(r(c)),\theta(c))\|
            \end{align*}
            and identically to the previous case, we apply \Cref{lem:cosin} below and get $$\|(r(b)+f(b),\theta(b))-(r(c)+f(c),\theta(c))\|=\|T((r(b),\theta(b)))-T(r(c),\theta(c)))\|\ge2$$

     \end{enumerate}
\end{proof}

\begin{claim}\label{fact:r1r2r3}
    Let $0<s\le r'\le r$ and $\theta_s,\theta_r\in[0,2\pi]$ then $\|(s,\theta_s)-(r',\theta_r)\|\le\|(s,\theta_s)-(r,\theta_r)\|$
\end{claim}
\begin{proof}
    Let $\theta=\theta_s-\theta_r$.
    \begin{align*}
        &\|(s,\theta_s)-(r,\theta_r)\|^2-\|(s,\theta_s)-(r',\theta_r)\|^2\\
        &=s^2+r^2-2sr\cos(\theta)-s^2-r'^2+2sr'\cos(\theta)\\
        &=(r-r')(r+r')-2s(r-r')\cos(\theta)=(r-r')(r+r'-2s\cos(\theta))\\
        &\ge(r-r')(r+r'-2s)>0
    \end{align*}
\end{proof}

\begin{claim}\label{lem:cosin}
    Let $f(x)=kx+K$, $k\ge-1$, and $K\ge0$. 
    Consider a transformation $T:\reals^2\setminus\{0\}\to\reals^2$ such that in the polar coordinate system we have $T((x,\theta))=(x+f(x),\theta)$. 
    Then for any two points $p_1,p_2\in\reals^2\setminus\{0\}$, it holds that $\|T(p_1)-T(p_2)\|\ge(1+k)\|p_1-p_2\|$.
\end{claim}
\begin{proof}
        Denote $p_1=(r_1,\theta_1)$ and $p_2=(r_2,\theta_2)$. By the law of cosines,
        \begin{align*}
            \|T(p_1)&-T(p_2)\|^2=\\        
            &=(r_1+f(r_1)-r_2-f(r_2))^2+2(r_1+f(r_1))(r_2+f(r_2))(1-\cos(\theta_1-\theta_2))\\
            &=(1+k)^2(r_1-r_2)^2+2((1+k)r_1+K)((1+k)r_2+K)(1-\cos(\theta_1-\theta_2))\\
            &\ge(1+k)^2(r_1-r_2)^2+2((1+k)r_1)((1+k)r_2)(1-\cos(\theta_1-\theta_2))\\
            &=(1+k)^2((r_1-r_2)^2+2(r_1r_2)(1-\cos(\theta_1-\theta_2)))=(1+k)^2\|p_1-p_2\|^2
        \end{align*}
\end{proof}

\subsection{The algorithm}
As mentioned at the beginning of this section, we run a version of \Cref{thm:algo_length_and_runtime} where instead of using the clearance paths defined in \Cref{sec:clearance_path} in step 5 of the algorithm, we use the clearance paths defined in this section.
The input of the algorithm is then $\W$, $S$, $T$, $0< \eps\le 1$ and $0<\alpha\le 1$ such that \Cref{con:boundary_seperation_to_clear_R,con:robot_seperation_to_ensure_R,con:robot_separation_4-2eps,con:boundary_seperation_to_clear_interruptions,con:robot_seperation_to_clear_interruptions_improved} hold.

\begin{remark}\label{rem:zeros}
    We require both $\alpha$ and $\eps$ to be positive so that we do not need to handle division by zero separately in our claims.
    However, it is in fact easy to see that when $\eps=0$ or $\alpha=0$, we simply get $f_b(w)=\delta(w)$, and then all the robots are moving whenever there is an overlap. This yields an approximation factor of $O(m)$, as stated in \Cref{cor:opt_sep_bounds_eps1}.
\end{remark}

To bound the running time of this algorithm, we first need the following lemma.

\begin{lemma}\label{lem:computing_delta}
    A partition of $[0,1]$ into a set $R$ of $O(m+n)$ ranges can be computed in $O(m^{1.5}+m^{0.5}n)$ time, such that for each range $[w_1,w_2]\in R$ it holds that $\gamma_A[w_1,w_2]$ is either a line segment or an arc, and the closest point to $A$ in time $[w_1,w_2]$ does not change.
\end{lemma}
\begin{proof}
    Because in each iteration of \Cref{alg:eps_overlap_algo} we add exactly one disk to $\O$, the complexity of $\F$ is $O(m+n)$ and its boundary consist of straight line segments and arcs (of radius $1$ or $2-\eps$). Therefore, the geodesic path $\gamma_A$ also consist of straight line segments and arcs and its complexity is $O(m+n)$.
    Let $\hat{R}$ be a set of $O(m+n)$ ranges in $[0,1]$ such that for each range $[w_1,w_2]\in R$ it holds that $\gamma_A[w_1,w_2]$ is either a line segment or an arc.
    
    To compute $\delta(w)=\max\{0,2-\min_{p\in S\cup T}\{\|\gamma_A(w)-p\|\}$, we need to find the points of $S\cup T$ that interrupt each part of the path. This can be done using classic simplex range searching data structures over $S\cup T$ (ignoring $\O$), as in~\cite{Matousek1993RangeSearching}. The preprocessing time is $(n^{1+c})$ for an arbitrarily small fixed constant $c$, and for a range (line segment or arc) $R_i$ that has $K_i$ points from $S\cup T$ within distance 2 from it the query time is $O(\sqrt{m}+K_i)$.
    Denote $K=\sum_{i=1,\dots,|R|}K_i$, so the total running time is $O(m^{1.5}+m^{0.5}n+K)$.
    After finding the interrupting points, we can partition each line segment or arc of $\gamma_A$ into smaller parts according to the closest point, e.g. by using a Voronoi diagram. This would take $O(K_i\log K_i)$ for the range $R_i$ and $O(K\log K)$ in total.

    Therefore, in order to bound the running time we need to bound $K$, the total number of interrupting points.
    Consider an arc piece of $\gamma_A$. Observe that it has radius $1$ if the robot turns around a reflex vertex of the workspace, and radius $2-\eps$ if the robot turns around a disk obstacle that the algorithm removed from $\W$. 
    Because the points of $S\cup T$ have distance $2+\alpha\eps$ form each other, there can be only a constant number of points interrupting each such arc piece. 
    To bound the number of points interrupting line segment pieces, we count separately the number of interruptions ``near'' the endpoints of the segments and in the ``interior'' of the segments. 
    Consider a segment piece $\gamma_A[w_1,w_2]$.
    Because the distance between every two points in $S\cup T$ is at least 2, there is only a constant number of points from $S\cup T$ that interrupt $\gamma_A[w_1,w_2]$ and lie within distance $10$ from one of its endpoints. 
    Thus, in total, the number of points that from $S\cup T$ that interrupt any segment piece is $O(m+n)$.
    Now consider a point $p\in S\cup T$ that interrupts $\gamma_A[w_1,w_2]$ and lies at distance at least $10$ from one of its endpoints. Since $\gamma_A$ is non crossing (otherwise it could be shorten), there are at most $2$ segments of $\gamma_A$ for which $p$ can interrupt at their ``interior'' (i.e., lie at distance at least $10$ from one of the endpoints). So in total, the number of points that interrupt a segment piece in its interior is $O(m)$.
    We conclude that $K=O(m+n)$, and the lemma follows.
\end{proof}

The following theorem is an adaptation of \Cref{thm:algo_length_and_runtime}.
Notice that the only change is the size of the solution, which now depends on $\alpha$.

\begin{restatable}{theorem}{improvedAlgoLength}\label{thm:improved_algo_length_and_runtime}
    Let $\eps,\alpha\in(0,1]$. When using the clearance paths from \Cref{sec:clearance_path2} in step 5 of \Cref{alg:eps_overlap_algo}, this algorithm runs in $O(m^4+m^2n^2)$ time, and returns a solution $\Gamma^*$ such that $|\Gamma^*|=O(\alpha^{-2}\OPT)$, where $\OPT$ is the size of an optimal solution.
\end{restatable}
\begin{proof}
   First, notice that since we only modified the definition of the clearance paths, the only modifications that we need to do in the proof of \Cref{thm:algo_length_and_runtime} are in the analysis of the clearance paths, i.e. their length and computation time.
   
   The total sum of distances traveled by robots when moving directly from their starting positions to their target positions remain unchanged, namely, $|OPT|+(4-2\eps)m$.
   To bound the extra distance traveled by the clearance paths, 
   let $b$ be the starting/target position occupied by some robot $B$. Consider a maximal subpath $\gamma_A(w_1,w_2)$ where $f_b(w)>0$ for every $w\in[w_1,w_2]$. In this time range, the robot $B$ traveled according to $\tilde\gamma_B[w_1,w_2]$, and by \Cref{clm:clearance_path_in_disk_improved} we have $\tilde\gamma_B[w_1,w_2]\le\sqrt{2}\gamma_A[w_1,w_2]$.
   Observe that at any given time $w$, the set of points $p$ in the plane for which $f_p(w)>0$ is a disk of radius $2+\frac{2\delta(w)}{\alpha\eps} \le2+\frac{2}{\alpha}$.
   Therefore, the number of robots that can be moved by $\gamma_A$ at any given time is bounded by the number of pairwise disjoint unit disks that can fit within such a disk, which is $O(\alpha^{-2})$.
   Proceeding similarly to the proof of \Cref{thm:algo_length_and_runtime}, by adding the extra distance traveled in the clearance paths in all the iterations we get that the size of the output solution is $|\Gamma^*|=O(\alpha^{-2}\OPT)$.

   As for the running time, again, we only need to analyze how the new clearance paths computation effects it. 
   To compute the clearance paths that correspond to the path $\gamma_A$ in a given iteration of the algorithm, we first apply \Cref{lem:computing_delta} to obtain a partition of $[0,1]$ into a set $R$ of ranges such that for each range $[w_1,w_2]\in R$ it holds that $\gamma_A[w_1,w_2]$ is either a line segment or an arc, and the closest point to $A$ in time $[w_1,w_2]$ does not change.
   Then, for each range $[w_1,w_2]\in R$ and for each point $p\in S\cup T$, the complexity of the function $f_p(w)$ is constant, and thus it can be computed in constant time.
   Since the number of ranges is $O(m+n)$, this takes $O(m+n)$ for each clearance path.
   By \Cref{lem:computing_delta}, the set $R$ can be computed in $O(m^{1.5}+m^{0.5}n)$time. So for $m$ robots in $m$ iterations, the total running time for computing the clearance paths is $O\left(m^3+m^2n\right)$ which is absorbed in the $O\left(m^4+m^2n^2\right)$ from the proof of \Cref{thm:algo_length_and_runtime}.
\end{proof}

We conclude:
\begin{theorem}\label{thm:opt_sep_bounds_improved}
    Let $\eps,\alpha\in(0,1]$. Given a polygonal workspace $\W$ with $n$ edges, a set $S$ of $m$ starting positions and a set $T$ of $m$ target positions, there is an algorithm that computes in $\Tilde{O}(m^4+n^2m^2)$ time a collision-free motion plan of length $O(\alpha^{-2}\OPT)$, where $\OPT$ is the length of an optimal solution, or reports that such a motion plan does not exist, assuming that $\rho=\max\{4-2\eps,2+\alpha\eps\}$ and $\omega=\max\{1+\eps,\sqrt{(3-\sqrt{3}-\eps)^2+1}\}$.
\end{theorem}

\subsection{Optimizing the separation bounds}\label{sec:optimzing_bounds2}
We now have two parameters, $\alpha$ and $\eps$, to optimize over \Cref{con:boundary_seperation_to_clear_R,con:robot_seperation_to_ensure_R,con:robot_separation_4-2eps,con:boundary_seperation_to_clear_interruptions} and \Cref{con:robot_seperation_to_clear_interruptions_improved}.

For $\frac{2}{3}\le\eps\le 1$, by choosing $\alpha=\frac{2}{\eps}-2$ we get that \Cref{con:robot_separation_4-2eps} and \Cref{con:robot_seperation_to_clear_interruptions_improved} are equal, and thus $\rho=4-2\eps$ and $\omega=1+\eps$ are sufficient for \Cref{alg:eps_overlap_algo} to work. In fact, we can get $\rho=4-2\eps$ and $\omega=1+\eps$ for any $\frac{15-6\sqrt{3}}{13}\le\eps\le 1$, because in the range $\frac{15-6\sqrt{3}}{13}\le\eps< \frac23$ we have $4-2\eps>2+\eps$.

In the case of $\eps=1$, we choose $\alpha=0$ (see \Cref{rem:zeros}), and thus in step 5 of \Cref{alg:eps_overlap_algo} all the robots may move in order to clear the way. Therefore, by \Cref{clm:clearance_path_in_disk_improved}, in this case the length of the solution becomes $O(m\cdot \OPT)$. We conclude:

\begin{corollary}\label{cor:opt_sep_bounds_eps1}
    Given a polygonal workspace $\W$ with $n$ edges, a set $S$ of $m$ starting positions and a set $T$ of $m$ target positions, there is an algorithm that computes in $\Tilde{O}(m^4+n^2m^2)$ time a collision-free motion plan of length $O(m\cdot\OPT)$, where $\OPT$ is the length of an optimal solution, or reports that such a motion plan does not exist, assuming that $\rho=2$ and $\omega=2$.
\end{corollary}

In \Cref{tab:separation2}, we describe the results that we achieve for different values of $\eps$.

\begin{table}[h]
\centering
\begin{tabular}{|l|l|l|l|l|}
\hline
$\eps$                                & $\rho$         & $\omega$                          & Solution length               & Comments                                                                             \\ \hline\hline
$0$                                   & $4$            & $\sqrt{13-6\sqrt{3}}\approx1.614$ & $\OPT+4m$                     & monotone                                                                             \\ \hline
$(0,\frac{15-6\sqrt{3}}{13})$         & $4-2\eps$      & $\sqrt{(3-\sqrt{3}-\eps)^2+1}$                          & $O(\OPT)$                     & weakly-monotone                                                                      \\ \hline
$\frac{15-6\sqrt{3}}{13}\approx0.354$ & $\approx3.291$ & $\approx1.354$                    & $O(\OPT)$                     & \begin{tabular}[c]{@{}l@{}}min obstacles-separation \\ weakly-monotone\end{tabular} \\ \hline
$(\frac{15-6\sqrt{3}}{13},\frac23)$   & $4-2\eps$      & $1+\eps$                          & $O(\OPT)$                     & weakly-monotone                                                                      \\ \hline
$\frac23$                             & $2\frac{2}{3}$ & $1\frac{2}{3}$                    & $O(\OPT)$                     & weakly-monotone     \\ \hline
$(\frac{2}{3},1)$                     & $4-2\eps$      & $1+\eps$                          & $O((\frac{1}{1-\eps})^2\cdot\OPT)$ &                                                                                      \\ \hline
$1$                                   & $2$            & $2$                               & $O(m\cdot\OPT)$                 & min robots-separation                 \\ \hline
\end{tabular}
\caption{The resulting algorithm for the different values of $\eps$.}
\label{tab:separation2}
\end{table}

Note that the separation bound of $\rho=2$ in \Cref{cor:opt_sep_bounds_eps1} above is tight for monochromatic robots-separation. However, the bichromatic robots-separation can potentially be $0$. On the other hand, below we provide lower bounds showing that obstacles-separation of at least $1.5$ is sometimes needed for unlabeled MRMP (see \Cref{sec:lower_bounds}).
This raises the question of what is the minimum obstacles-separation for which we can remove the robots-separation assumption completely.

\section{Near Tight Separation Bounds for Labeled MRMP}\label{sec:labeled}
For $\rho=2$ and $\omega=2$, we can use the ideas from \Cref{sec:improved} in order to give an approximate solution for Labeled MRMP (and in fact, show that a solution for Labeled MRMP always exists when $\omega=2$ and the free space is connected).
Given a set $S=\{s_1,\dots,s_m\}$ of starting points and a set $T=\{t_1,\dots,t_m\}$ of target points in a workspace $\W$, the goal is to compute a motion plan for the robots that results in $s_i$ placed on $t_i$ for every $i\in[m]$.

The idea is to treat all the points in $S\cup T$ as obstacles (of 0 extent), and then compute for each $i$ the geodesic path $\hat{\gamma}_i$ from $s_i$ to $t_i$ in this modified workspace $\widehat{W}$. 
Next, we argue that a path $\hat{\gamma}_i$ exists in $\widehat{\W}$ for every $i\in[m]$, and that it is not much longer then the geodesic path $\gamma_i$ between $s_i$ and $t_i$ in the original workspace $\W$. 
Consider a unit disk $D$ around a point from $S\cup T$ that intersect the path $\gamma_i$, and denote by $a(D)$ the smaller arc created by this intersection. We modify $\gamma_i$ by replacing each subpath of $\gamma_i$ that intersect a unit disk $D$ around a point from $S\cup T$ by the arc $a(D)$. This modification results in a path $\hat{\hat{\gamma}}_i$ which has length at most $\frac{2\pi r}{2r}=\pi$ times the length of $\gamma_i$. Moreover, $\hat{\hat{\gamma}}_i$ is in the free space of $\widehat{\W}$: every point on the boundary of $D$ is at distance at least 1 from the boundary of $\W$ (because $\omega=2$), and at distance at least 1 from each point in $S\cup T$ (because $\rho=2$). Therefore a path between $s_i$ and $t_i$ in $\widehat{\W}$ exists, and by the minimality of $\hat{\gamma}_i$, we have $|\hat{\gamma}_i|\le |\hat{\hat{\gamma}}_i|\le \pi\cdot|\gamma_i|$.

Notice that each of the paths $\hat{\gamma}_i$ is a $1$-interrupted path in the original workspace $\W$ (it is not $1$-blocked by any other start or target position). We can therefore iteratively run steps 5 and 6 from \Cref{alg:eps_overlap_algo} for each path $\hat{\gamma}_i$, i.e., compute the clearance paths as described in \Cref{sec:improved}. The resulting set of paths is collision free and by the argument in \Cref{sec:improved} have a total length of $O(m\cdot \OPT)$ where $\OPT$ is the length of an optimal solution.

As we argue in the proof of \Cref{thm:improved_algo_length_and_runtime}, the total running time for computing all clearance paths is $\Tilde{O}(m^3+m^2n)$.
We therefore conclude with the following theorem.
\begin{theorem}\label{thm:labeled_tight}
    Given a polygonal workspace $\W$ with $n$ edges, a set $S$ of $m$ \textbf{labeled} starting positions and a set $T$ of $m$ \textbf{labeled} target positions, there is an algorithm that computes in $\Tilde{O}(m^3+m^2n)$ time a collision-free motion plan of length $O(m\OPT)$, where $\OPT$ is the length of an optimal solution (for Labeled MRMP), or reports that such a motion plan does not exist, assuming that $\rho=2$ and $\omega=2$.
\end{theorem}

In fact, \Cref{thm:labeled_tight} above shows that if $\rho=2$, $\omega=2$ and for every $i$ it holds that $s_i$ and $t_i$ are in the same connected component of the free space, a solution to Labeled MRMP always exists.
In section \Cref{sec:lower_bounds} below, we show that when $\omega<2$, a solution does not always exist in the labeled variant, even if
the free space consists of a single connected component. 
Therefore the separation assumption on $\omega$ that we achieve in this section is tight for the labeled variant. Note that it may still be possible to improve the robots-separation by distinguishing chromatic and bichromatic separation as in~\cite{BBBBFHKOS22}.

\section{Limitations of Monotone and Weakly-Monotone Motion Plans}\label{sec:weakly_monotone_limitations}
In this section we present two examples that demonstrate the limitations of (weakly) monotone motion plans.
In the first example the obstacles-separation is $\approx1.614$, and we show that no monotone motion plan exists. In the second example the obstacles-separation is $\approx1.354$ and we show that no weakly-monotone motion plan exists.
These examples show that in some situations, monotone or even weakly monotone techniques cannot yield a solution when the obstacles-separation is smaller than those presented in this paper, and therefore a completely different approach is needed for algorithms with lower obstacles-separation assumptions.

\subsection{Monotone motion plans}
Recall that in a monotone motion plan with $\omega=\sqrt{13-6\sqrt{3}}\approx1.614$, each start/target position must be at least $\omega$ away from every obstacle, and the robots move one by one from their starting positions to their target positions, each fully completing its path and never moving again. 

Consider \Cref{fig:monotone_low_bound}. 
The point $o$ is the origin, and the target $t_1$ is located at $(0,-2+\delta_1)$ for some arbitrarily small $\delta_1>0$.
The orange rectangle is $\R(o,t_1)$.
The vertex $v_1$ is located at $(0,1)$, and the dashed gray disk centered at $v_1$ has radius $2$. 
The points $v_2,v_3$ are located at distance exactly $2$ from $v_1$, slightly below the intersection points of $D_2(v_2)$ and $\R(o,t_1)$ (the orange points), whose coordinates are $(\pm1,1-\sqrt3)$. 
Note that the distance between these orange points is exactly $2$, and that the distance between $v_2$ and $v_3$ is strictly smaller than $2$.
The distance between $t_1$ and each of the orange points is exactly $\sqrt{1-(1-\sqrt3+2-\delta_1)^2}=\rho_{\delta_1}$, which approaches $\rho_0=\sqrt{13-6\sqrt{3}}$ when $\delta_1$ approaches $0$.
We can thus pick $\delta_1>0$ and the points $v_2,v_3$ such that the distance between $t_1$ and each of $v_2,v_3$ is $\rho_0-\delta_2$ for any arbitrarily small $\delta_2>0$.
The setting for $t_2$ is symmetric.

Now, observe that $s_1,t_1$ are in one connected component of $\F$ while $s_2,t_2$ are in another one (because the distance between $v_2$ and $v_3$ is strictly smaller than $2$). 
In a monotone motion plan, one of the robots will move first to a target, and then it will stay there. Assume w.l.o.g. that the robot on $s_1$ goes first and reaches $t_1$. To reach $t_2$, the robot positioned on $s_2$ must pass through $\overline{v_1v_2}$ and $\overline{v_1v_3}$, and therefore the trace of its path must cover the arc of $D_2(v_1)$ between $v_2$ and $v_3$. Since $\|t_1-v_1\|<2$, the robot positioned on $t_1$ would always intersect the trace of such path, so $s_2$ cannot reach $t_2$.

We conclude that any algorithm that returns a monotone motion plan will not be able to produce an answer to this instance, even though it is easy to see that a feasible solution exists. We note that \cite{ABHS15} and \cite{BBBBFHKOS22} present monotone algorithms with $\omega=1$, but they work only in simple polygons.

\begin{figure}[h!]
    \centering
        \includegraphics[page=2, scale=0.7]{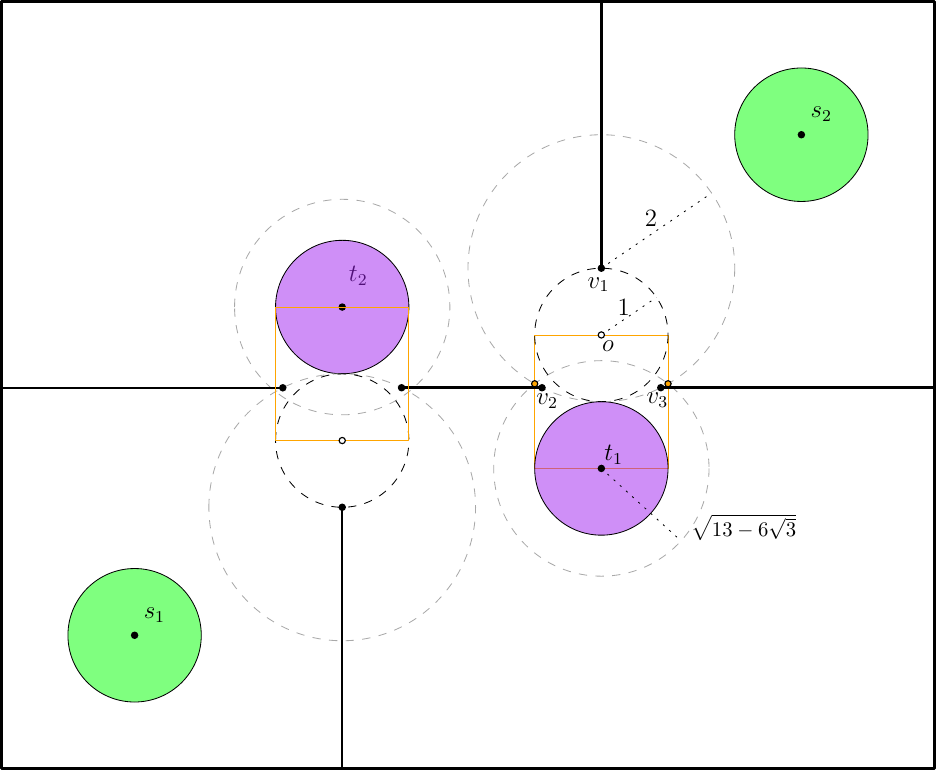}
    \caption{A instance of MRMP with $\omega\approx1.614$, in which a monotone solution does not exist.}
    \label{fig:monotone_low_bound}
\end{figure}

\subsection{Weakly-monotone motion plans} 
Recall that in a weakly-monotone motion plan with respect to radius $r$, each start/target position must be contained in disk of radius $r$, empty of obstacle and other robots. The robots move one by one from their starting positions to their target positions, while other robots must remain inside a disk of radius $r$ centered either at another start or target position.

Let $\eps=\frac{15-6\sqrt{3}}{13}\approx0.354$ and consider \Cref{fig:monotone_low_bound}.
The point $o_1$ is the origin, and $o_2=(0,-2)$. The orange rectangle is $\R(o_1,o_2)$.
The vertex $v_1$ is located at $(0,1)$, and the dashed gray disk centered at $v_1$ has radius $2$. 
The points $v_2,v_3$ are located at distance exactly $2$ from $v_1$, slightly below the intersection points of $D_2(v_2)$ and $\R(o_1,o_2)$ (the orange points), whose coordinates are $(\pm1,1-\sqrt3)$. Note that the distance between these orange points is exactly $2$, and that the distance between $v_2$ and $v_3$ is strictly smaller than $2$.

We place the target $t_1$ at $(0,-2+\eps)$, so the distance between $t_1$ and each of the orange points is exactly $\sqrt{1-(1-\sqrt3+2-\eps)^2}=\rho_\eps$. One can verify that for $\eps=\frac{15-6\sqrt{3}}{13}$ we have $\rho_\eps=1+\eps$ (this is how we chose $\eps$ in case (ii) of \Cref{thm:opt_sep_bounds}). This means that $\D_{1+\eps}(t_1)$ (marked in dashed gray) contains $v_2,v_3$, and therefore the largest disk centered at $t_1$ that does not intersect the boundary of the polygon has radius $1+\eps-\delta$, for some arbitrarily small $\delta>0$. Since $\|t_1-v_1\|=3-\eps$, we get that $\D_{1+\eps-\delta}(t_1)$ is contained in $\D_{4-\delta}(v_1)$.
The setting for $t_2$ is symmetric.

Now, as in the previous example, $s_1,t_1$ are in one connected component of $\F$ while $s_2,t_2$ are in another one (because the distance between $v_2$ and $v_3$ is strictly smaller than $2$). 
In a weakly monotone motion plan, one of the robots will move first to a target, and then it will stay in the disk of radius $1+\eps-\delta$ around it (this is the largest disk centered at the target and does not intersect the boundary of the polygon). Assume w.l.o.g. that the robot on $s_1$ goes first. When it reaches $t_1$, it can still move inside $\D_{1+\eps-\delta}(t_1)$. However, to reach $t_2$, the robot positioned on $s_2$ must pass through $\overline{v_1v_2}$ and $\overline{v_1v_3}$, and therefore the trace of its path must cover the arc of $D_2(v_1)$ between $v_2$ and $v_3$. Since $\D_{1+\eps-\delta}(t_1)$ is contained in $\D_{4-\delta}(v_1)$, the robot positioned on $t_1$ would always intersect the trace of such path, so $s_2$ cannot reach $t_2$.

We conclude that any algorithm that plans monotone motion plan will not be able to produce an answer to this instance although it is easy to see that a feasible solution exists. 
Recall that \cite{AGHT23} present a weakly-monotone algorithm in which $r=2$, and the start/target positions do not have to lie at the center of the revolving area. However, we note that it is possible to modify the example in \Cref{fig:monotone_low_bound} to force the disk $\D_{1+\eps-\delta}(t_1)$ to be the revolving area, even when the target position is not $t_1$, by blocking it with the boundary of the polygon.

\begin{figure}[h!]
    \centering
        \includegraphics[page=3, scale=0.7]{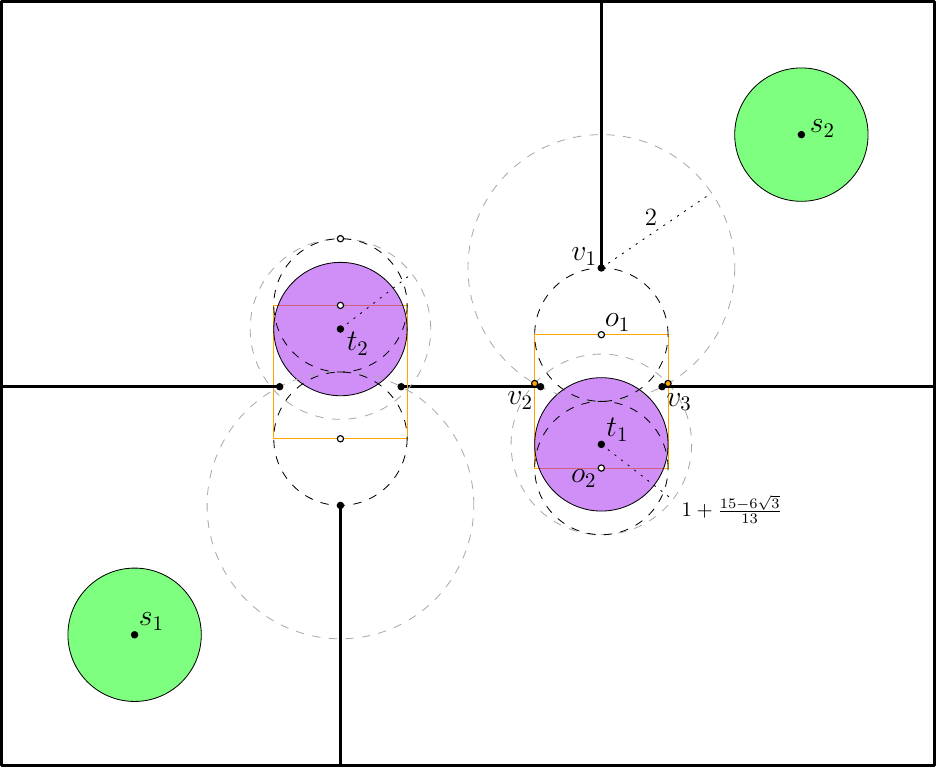}
    \caption{A instance of MRMP with $\omega\approx1.354$, in which a weakly-monotone solution does not exist.}
    \label{fig:weakly_monotone_low_bound}
\end{figure}

\section{Lower Bounds}\label{sec:lower_bounds}
In this section we present lower bounds for the amount of obstacles-separation needed so that a solution to the MRMP problem always exists.

First, in the unlabeled case, we show that obstacles-separation of at least $1.5$ is sometimes needed.

\begin{figure}[h!]
    \centering
    \includegraphics[page=1,scale=0.8]{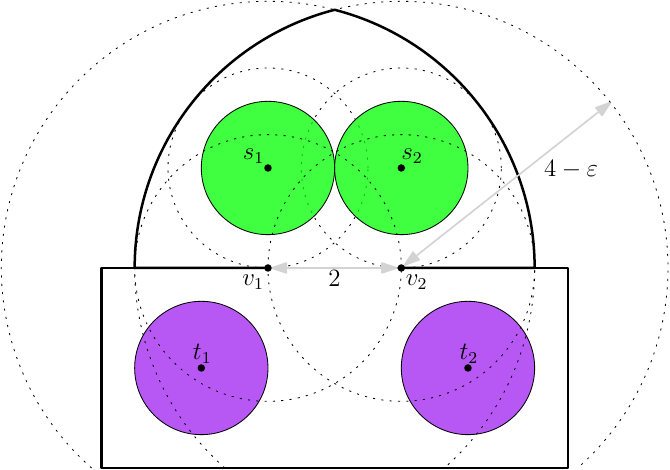}
    \caption{An instance of unlabeled MRMP with obstacles-separation of $1.5-\eps$ for which no solution exists.}
    \label{fig:hourglass1_5}
\end{figure}
\begin{lemma}
    For $\omega<1.5$ a solution does not always exist in the unlabeled variant, even if the free space consists of a single connected component containing two start and two target positions, and the bichromatic robots-separation is arbitrarily large.
\end{lemma}
\begin{proof}
    Consider \Cref{fig:hourglass1_5}. The positions $s_1,s_2$ are at distance exactly $2$ from each other. The dotted disks centered at $s_1$ and $s_2$ are of radius $1.5-\eps$, for some small $\eps>0$. The point $v_1$ is directly below $s_1$ and $v_2$ is directly below $s_2$, so the distance between $v_1$ and $v_2$ is also exactly $2$. The upper ``room'' is constructed using two arcs, the one on the right belongs to the disk of radius $4-\eps$ centered at $v_2$, and the one on the left belongs to a disk of radius $4-\eps$ centered at $v_1$. The positions $t_1,t_2$ lie in the bottom ``room'', below the opening between $v_1$ and $v_2$.
    Observe that the free space is a single connected region, as the narrow passage between $v_1$ and $v_2$ is of width $2$. This implies that there exists a straight-line path within the free space connecting the two rooms of the polygon.
    
    Now, in order for a robot $A$ positioned on $s_1$ to reach the bottom room, the robot $B$ lying on $s_2$ must step aside. Since the distance between $v_1$ and $v_2$ is exactly $2$, when $A$ passes through the narrow passage between $v_1$ and $v_2$, its boundary must touch both $v_1$ and $v_2$ at the same time. Consider the first time when the the boundary of $A$ touches $v_1$. At that time, $A$ must be completely contained in $D_2(v_1)$. However, notice that the upper room is contained in a disk of radius $4-\eps$ around $v_1$, and therefore the robot $B$ must intersect $D_2(v_1)$. 
    If $B$ is located to the right of $A$, it prevents $A$ from reaching the passage. On the other hand, $A$ and be cannot switch sides, i.e., $B$ will always be to the right of $A$, because by symmetry $A$ prevents $B$ from passing above it.
    By symmetry, $B$ cannot reach the passage as well, and therefore there is no solution to this instance for $\eps>0$.

    Note that for $\eps=0$, the robots can rotate in a disk of radius $2$ inside the upper room, so that when $B$ reaches the top point in the disk, $A$ can enter the narrow passage between $v_1$ and $v_2$.
\end{proof}

In the labeled case, we observe that obstacles-separation of at least $2$ is sometimes necessary (for any value of $\rho$).
\begin{figure}[h!]
    \centering
    \includegraphics[page=2,scale=0.7]{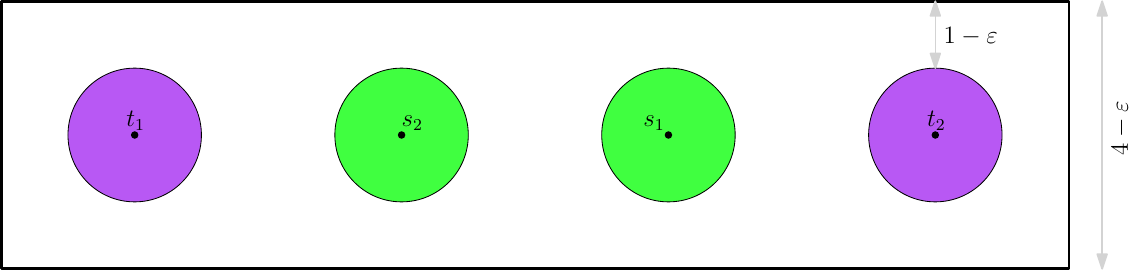}
    \caption{An instance of labeled MRMP with obstacles-separation of $2-\eps$ for which no solution exists.}
    \label{fig:strip_lower_bound}
\end{figure}
\begin{lemma}
    For $\omega<2$ a solution does not always exist in the labeled variant, even if the free space consists of a single connected component containing two start and two target positions, and $\rho$ is arbitrarily large.
\end{lemma}
\begin{proof}
Consider \Cref{fig:strip_lower_bound}. The workspace is a long strip-shaped polygon of width $4-\eps$ for some arbitrarily small $\eps>0$. The start and target points are aligned in the following order along the strip: $t_1,s_2,s_1,t_2$. The distance between any start/target position to the boundary is $2-\eps$.
In order for a robot $A$ positioned on $s_1$ to move to $t_1$, it needs to switch sides with a robot $B$ positioned on $s_2$. However, this is not possible: assume by contradiction that such a switch is possible, then at some point the $x$-coordinates of $A$ and $B$ should be equal, but this is not possible because the width of the strip is $4-\eps$.
\end{proof}

We conclude that the best possible obstacles-separation (without further assumptions) that we can hope for is $1.5$ in the unlabeled case, and $2$ in the labeled case.

\bibliography{refs}

\appendix

\section{Geodesics and blocking positions in the free space}\label{apx:geodesics}
In this section we provide some very useful geometric properties of geodesic paths in the free space, which allow us to set the constraints under which a blocking robot can enter a geodesic path via a straight line segment connecting it to the path.

Notice that when $\W$ is a polygonal environment, the boundary of $\F$ consists of straight line segments and arcs of unit circles (see, e.g., \cite{Wein07}). Moreover, if $\F$ contains a point $x$ such that $D_1(x)$ touches a reflex vertex $v$ of $\W$, then the boundary of $\F$ contains an arc of $D_1(v)$.

A geodesic path $\gamma$ in $\F$ also consists of straight line segments and arcs of unit circles, where the straight line segments are tangent to circular arcs on the boundary of $\F$ that correspond to reflex vertices of $\W$. Thus for any point in the interior of a geodesic path there is a single tangent line. 
Observe that since $\gamma$ is a shortest path, if a point $\gamma(t)$ for some $0\le t\le 1$ lies in the interior of $\F$, then it is locally a straight line segment. In other words, we have the following observation:
\begin{observation}\label{obs:geodesic_paths}
    Let $\gamma: [0, 1]\rightarrow \F$ be a geodesic path, and $0\le t\le 1$. 
    If $\gamma[t-\delta,t+\delta]$ is not a straight line segment for any $\delta>0$, then $\gamma(t)$ is on the boundary of $\F$, and therefor $\D_1(\gamma(t))$ touches the boundary of $\W$ at a point $x$. Moreover, $\overline{\gamma(t)x}$ is orthogonal to the tangent of $\gamma$ at the point $\gamma(t)$.
\end{observation}

\begin{figure}[h!]
    \centering
    \includegraphics[page=2,scale=1]{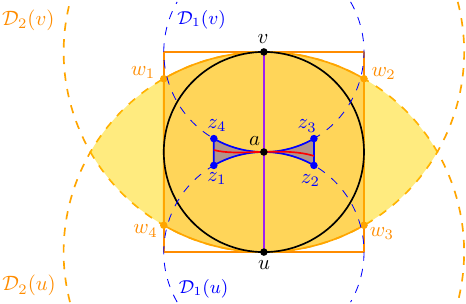}
    \caption{The region $L(a,\overline{uv})$ is in dark yellow. The funnels $F_1(a,\overline{uv}), F_2(a,\overline{uv})$ are in blue.}
    \label{fig:lens}
\end{figure}

\oldparagraph{The cut-out lens and the funnel ports.}
Consider a point $a$, and let $\overline{vu}$ be a diametric chord of $\D_1(a)$ (see \Cref{fig:lens}). Let $\S$ be the smallest square enclosing $\D_1(a)$ with two edges parallel to $\overline{vu}$ (its edges are of length $2$). Let $L(a,\overline{uv}):=\D_2(v)\cap\D_2(u)\cap \S$.
Denote by $w_1,w_2,w_3,w_4$ the vertices of $L(a,\overline{uv})$ such that $w_1,w_2$ are on the boundary of $D_2(u)$ and $w_3,w_4$ are on the boundary of $D_2(v)$. Let $z_1$ (resp. $z_2$) be the intersection point of $D_1(u)$ with the segment $\overline{w_1u}$ (resp. $\overline{w_2u}$). Similarly, let $z_3$ (resp. $z_4$) be the intersection point of $D_1(v)$ with the segment $\overline{w_3v}$ (resp. $\overline{w_4v}$).
Denote by $F_1(a,\overline{uv})$ (resp. $F_2(a,\overline{uv})$) the funnel defined by the arc $\overset{\frown}{az_1}$ of $D_1(u)$ (resp. the arc $\overset{\frown}{az_2}$ of $D_1(u)$), the arc $\overset{\frown}{az_4}$ of $D_1(v)$ (resp. the arc $\overset{\frown}{az_3}$ of $D_1(v)$), and the segment $\overline{z_1z_4}$ (resp. $\overline{z_2z_3}$).

Intuitively, the claim below states that the trace of any path that is contained in the hourglass $F_1(a,\overline{uv})\cup F_2(a,\overline{uv})$ and crosses $\overline{z_1z_4}$ and $\overline{z_3z_2}$ would ``cover'' the entire region $L(a,\overline{uv})$.
It follows from the fact that when rolling a unit disk on the boundary of $L(a,\overline{uv})$, we always cover a segment that connects the upper and lower arcs of the hourglass $F_1(a,\overline{uv})\cup F_2(a,\overline{uv})$.
\begin{restatable}{claim}{pathInFunnels}\label{clm:path_in_funnels}
    Let $\pi$ be a continuous path in $F_1(a,\overline{uv})\cup F_2(a,\overline{uv})$ that starts at a point on $\overline{z_1z_4}$ and ends at a point on $\overline{z_3z_2}$. Then $L(a,\overline{uv})\subseteq D_1(\pi)$.
\end{restatable}
\begin{proof} 
    We show that for any point in $L(a,\overline{uv})$ there exists a point on $\pi$ within distance $1$ from it.
    First, notice that since $F_1(a,\overline{uv})\cap F_2(a,\overline{uv})=\{a\}$, the path $\pi$ must include $a$, and thus the claim it true for any point $p\in D_1(a)$. We therefore need to prove the claim for points in $L(a,\overline{uv})\setminus D_1(a)$.

    Assume w.l.o.g. that $a=(0,0)$, $v=(1,0)$, $u=(-1,0)$ (see \Cref{fig:coordinates}). Then the vertices of $L=L(a,\overline{uv})$ are then $w_1=, w_2=(1,-1+\sqrt3), w_3=(1,-1-\sqrt3), w_4=(-1,-1-\sqrt3)$.
    One can also verify that $z_1=(-\frac12,-1+\frac{\sqrt{3}}{2}), z_2=(\frac12,-1+\frac{\sqrt{3}}{2}), z_3=(\frac12,1-\frac{\sqrt{3}}{2}), z_4=(-\frac12,1-\frac{\sqrt{3}}{2})$.

    \begin{figure}
        \centering
        \includegraphics[scale=1.5,page=7]{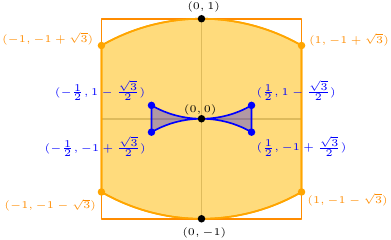}
        \caption{The coordinates assuming that $a=(0,0)$.}
        \label{fig:coordinates}
    \end{figure}

    Because $w_1$ is on the boundary of $D_2(u)$ we have $\|w_1-u\|=2$, and because $z_1$ is on the boundary of $D_1(u)$ we have $\|z_1-u\|=1$. Therefore $\|w_1-z_1\|=1$, so $w_1\in D_1(z_1)$. 
    By symmetric arguments, we have $w_i\in D_1(z_i)$ for every $1\le i\le 4$. Moreover, one can verify that $w_1\in D_1(z_4)$, $w_2\in D_1(z_3)$, $w_3\in D_1(z_2)$ and $w_4\in D_1(z_1)$.
    
    Consider the quadrilateral $Q_{1,4}=w_1z_4z_1w_4$, and notice that $Q_{1,4}\subseteq D_1(z_1)\cap D_1(z_4)$. Thus for any point $p\in Q_{1,4}$ we have $\overline{z_1z_4}\in D_1(p)$. Similarly, every point $p$ in the quadrilateral $Q_{2,3}=w_2z_3z_2w_3$ we have $\overline{z_2z_3}\in D_1(p)$. We conclude that every point in $Q_{1,4}\cup Q_{2,3}$ contains a point in $\pi$ (either the starting or the ending point).

    Finally, consider the region $T_1$ enclosed by the segments $\overline{w_1z_1},\overline{vz_1}$ and the arc $\overset{\frown}{w_1v}$. Notice that $D_1(z_4)$ contains $T_1$. 
    Now, because $T_1\subseteq D_2(u)$, for every point $p\in T_1$ we have $D_1(p)\cap D_1(u)\neq\emptyset$. Moreover, since $T_1$ is between the radii $\overline{w_1u}$ and $\overline{vu}$ of $D_2(u)$, the intersection $D_1(p)\cap D_1(u)$ must contain a point on the arc $\overset{\frown}{z_1a}$ of $D_1(u)$. We conclude that for every point in $T_1$ there is a segment $\overline{z_4p'}$ such that $p'$ is on $\overset{\frown}{z_1a}$, and $\overline{z_4p'}\subseteq D_1(p)$. Since $\pi$ is a continuous path in $F_1(a)\cup F_2(a)$ from a point on $\overline{z_1z_4}$ to a point on $\overline{z_3z_2}$, it must cross the segment $\overline{z_4p'}$, and therefore there is a point on $\pi$ within distance $1$ from $p$. By symmetric arguments, the claim is also true for the regions $T_2,T_3,T_4$ that correspond to $w_2,w_3,w_4$, respectively.

    The claim now follows because $L(a,\overline{uv})=D_1(a)\cup Q_{1,4}\cup Q_{2,3}\cup \bigcup_{i=1}^4 T_i$.
\end{proof}

Next, we use \Cref{clm:path_in_funnels} to show that given a point $a$ on a geodesic path $\gamma$, which is far enough from being an endpoint of $\gamma$, there is always a region $L(a,\overline{uv})$ that is contained in the trace of $\gamma$.

\begin{restatable}{claim}{theEyeIsInTheTrace}\label{clm:the_eye_is_in_the_trace}
    Let $\gamma: [0, 1]\rightarrow \F$ be a geodesic path, and let $a=\gamma(t)$ for some $0\le t\le 1$.
    Denote by $\overline{vu}$ the diametric chord of $\D_1(a)$ perpendicular to the tangent $\ell$ of $\gamma$ through $a$. If $\gamma(0),\gamma(1)\notin F_1(a,\overline{uv})\cup F_2(a,\overline{uv})$, then $L(a,\overline{uv})\subseteq\D_1(\gamma)$.
\end{restatable}
\begin{proof}
    We show that if $\gamma(0),\gamma(1)\notin F_1(a,\overline{uv})\cup F_2(a,\overline{uv})$, then $\gamma$ has a continuous subpath $\pi$ that starts at a point on $\overline{z_1z_4}$ and ends at a point on $\overline{z_2z_3}$, and $\pi$ is completely contained in $F_1(a,\overline{uv})\cup F_2(a,\overline{uv})$. Then by applying \Cref{clm:path_in_funnels} we get that $L(a,\overline{uv})\subseteq D_1(\pi)\subseteq D_1(\gamma)$.
    
    Assume w.l.o.g., as in the proof of \Cref{clm:path_in_funnels}, that $a=(0,0)$, $v=(0,1)$, and $u=(0,-1)$. Also assume w.l.o.g. that $\gamma$ crosses $\overline{uv}$ from left to right.

    \begin{figure}[h!]
        \centering
        \includegraphics[page=5,scale=1]{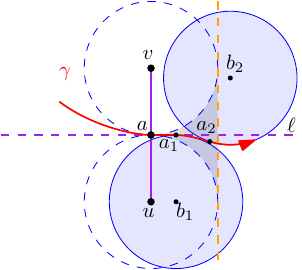}\hspace{1cm}
        \includegraphics[page=6,scale=0.7]{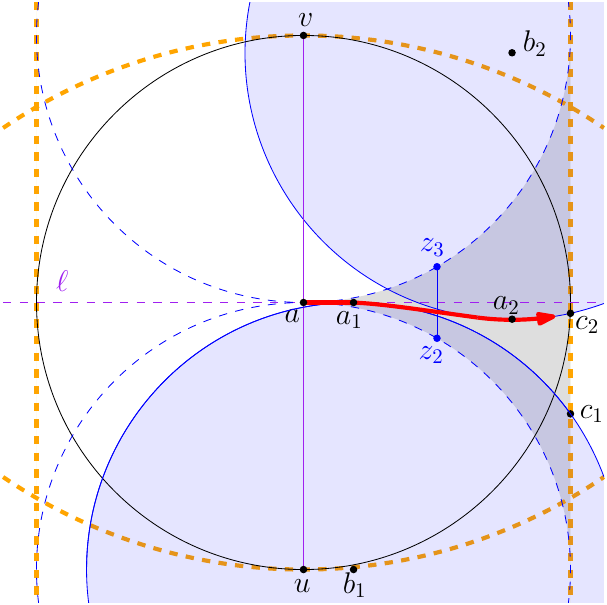}
        \caption{The funnel $F'_2$ is shaded in gray.}
        \label{fig:path_in_lens}
    \end{figure}

    Consider the funnel $F'_2$ enclosed by the line $x=1$ and the disks $D_1(u)$ and $D_1(v)$ (see \Cref{fig:path_in_lens}).
    We will show that $\gamma[t,1]$ does not cross the arcs of $F'_2$, and because $\gamma(1)\notin F_2(a,\overline{uv})$ this would imply that there exists a value $t\le t'\le 1$ such that $\gamma[t,t']$ is contained in $F_2(a,\overline{uv})$ and $\gamma(t')\in \overline{z_2z_3}$. The other part of the subpath of $\gamma$ in $F_1(a,\overline{uv})$ then follow by symmetric arguments. 

    Let $t_1\ge t$ be the largest value for which $\gamma[t,t_1]$ is a straight line segment. Because $\overline{uv}$ is perpendicular to $\ell$, and $\ell$ is also tangent to both $D_1(u)$ and $D_1(v)$, we get that $\gamma[t,t_1]\subseteq\ell$. If $t_1=1$, then clearly $\gamma[t,1]$ crosses $\overline{z_2z_3}$ and we are done. Otherwise, by \Cref{obs:geodesic_paths} the point $a_1=\gamma(t_1)$ must lie on the boundary of $\F$, and $\D_1(a_1)$ touches the boundary of $\W$ at a point $b_1$ which lies either exactly below or exactly above $a_1$.
    Assume w.l.o.g. that $b_1$ lies exactly below $\ell$. Because $b_1$ is on the boundary of $\W$, the path $\gamma$ does not intersect $\D_1(b_1)$, and the subpath of $\gamma$ after $a_1$ has to lie above the arc $\overset{\frown}{a_1c_1}$, where $c_1$ is the intersection point of the line $x=1$ and the boundary of $D_1(b_1)$.
    
    Since $\gamma$ is a geodesic path, the point $a_1$ has to be a tangency point of the line through $\overline{aa_1}$ (i.e., $\ell$) with the boundary of $\F$. Therefore, there exists some $\eps>0$ for which $\gamma[t_1,t_1+\eps]$ is monotonically decreasing. Now let $t_2\ge t_1$ be the largest value for which $\gamma[t_1,t_2]$ is monotonically decreasing. If $t_2=1$, then again we get that $\gamma[t,1]$ crosses $\overline{z_2z_3}$ and we are done. Otherwise, by \Cref{obs:geodesic_paths} the point $a_2=\gamma(t_2)$ must lie on the boundary of $\F$, and $\D_1(a_2)$ touches the boundary of $\W$ at a point $b_2$. Moreover, $b_2$ has to lie above $\ell$, and because $b_2$ is on the boundary of $\W$, the path $\gamma$ does not intersect $\D_1(b_2)$. We get that the subpath of $\gamma$ after $a_2$ has to lie below the arc $\overset{\frown}{a_2c_2}$, where $c_2$ is the intersection point of the line $x=1$ and the boundary of $D_1(b_2)$.

    Notice that $D_1(b_1)$ and $D_1(b_2)$ do not intersect, and therefore $c_2$ must lie above $c_1$. Consider the area enclosed by $\overset{\frown}{a_1c_1}$, $\overline{c_1c_2}$, $\overset{\frown}{a_2c_2}$, and $\gamma[t_1,t_2]$. This region contains the subpath of $\gamma$ from $a_1$ and until it crosses $\overline{c_1c_2}$. Since this region is also contained in the larger funnel $F'_2$, we get that $\gamma$ can only cross $F_2(a,\overline{uv})$ at some point on $\overline{z_2z_3}$.
\end{proof}

\oldparagraph{The magic radii.} Denote $\sigma_\eps=\sqrt{(3-\sqrt{3}-\eps)^2+1}$. This formula describes the obstacles-separation that is sufficient in order for a blocking robot to enter the path that it is blocking.

\begin{figure}[h!]
    \centering
    \includegraphics[page=1,scale=1]{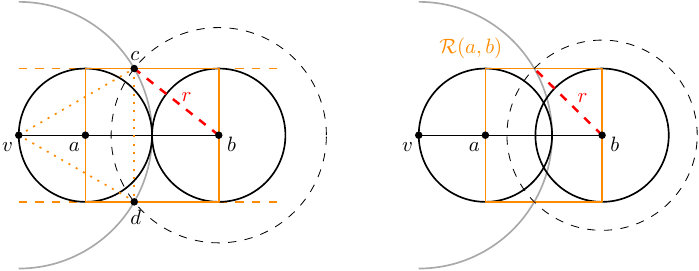}
    \caption{The rectangle $\R(a,b)$ is in orange.}
    \label{fig:free_rectangle}
\end{figure}

For two points $a,b$ such that $\D_1(a)\cap\D_1(b)\neq\emptyset$, denote by $\R(a,b)$ the rectangle whose edges are the two diametrical chords of $\D_1(a)$ and $\D_1(b)$ perpendicular to $\overline{ab}$, and the two line segments connecting them on the joint tangents to $\D_1(a)$ and $\D_1(b)$ (See \Cref{fig:free_rectangle}).

\begin{restatable}{claim}{freeRectangle}\label{clm:free_rectangle}
      Let $a,b$ be $\eps$-overlapping, and let $v$ be the intersection point of the line through $a$ and $b$ with the boundary of $\D_1(a)$, which is farther from $b$. Then $\R(a,b)\subset \D_2(v)\cup\D_{\sigma_\eps}(b)$.
\end{restatable}
\begin{proof}
    Assume w.l.o.g. that $v,a,b$ are on the $x$-axis, and that $v=(0,0)$, see \Cref{fig:free_rectangle}. Then $a=(1,0)$, and $b=(1+\|b-a\|,0)$.
    
    Consider the two tangents to $\D_1(a)$ that are parallel to $\overline{ab}$, one above and one below $\overline{ab}$, and denote by $c,d$ the top and bottom intersection points of these tangents with $\D_2(v)$, respectively. Clearly, $D_2(v)$ covers the part of $R(a,b)$ to the left of $\overline{cd}$, and we now show that $\D_r(b)$ covers the part of $R(a,b)$ to the right of $\overline{cd}$.
    
    The distance between the tangents is $2$, and therefore $\|c-d\|=2$. Since $\|v-c\|=\|v-d\|=2$, the triangle $\triangle vcd$ is an equilateral triangle, and therefore $c=(\sqrt{3},1)$ and $d=(\sqrt{3},-1)$. We have 
    \begin{align*}
        \|b-c\|=\|b-d\|&=\sqrt{(1+\|b-a\|-\sqrt{3})^2+(0-1)^2}\\
        &=\sqrt{(1-\sqrt{3}+\|b-a\|)^2+1}=\sqrt{(3-\sqrt{3}-(2-\|b-a\|))^2+1}. 
    \end{align*}
    The claim follows as $\eps=\max\{0,2-\|b-a\|\}$.
\end{proof}

Note that if for two points $a,b$ we have $\D_1(a)\cap\D_1(b)\neq\emptyset$, then $\|b-a\|<2$, and since $\sigma_\eps$ is monotonically decreasing when $\eps$ is increasing, by \Cref{clm:free_rectangle} we get that $R(a,b)$ is contained in $\D_2(v)\cap\D_{\sigma_0}(b)$ (where $\sigma_0=\sqrt{13-6\sqrt{3}}$). This is where the bound conjectured by Solovey et al.~\cite{SYZH15} is coming from, however, note that it requires the blocking robot to enter the path through a locally closest point (and not as defined in~\cite{SYZH15}), as we will clarify later.

We are now ready to set the constraints which allow a blocking robot to enter the path that it is blocking. We have one constraint on the obstacles-separation, and another on the robots-separation.
\constraintBoundary*
\constraintRobots*

\begin{figure}[h!]
    \centering
    \includegraphics[page=4,scale=1]{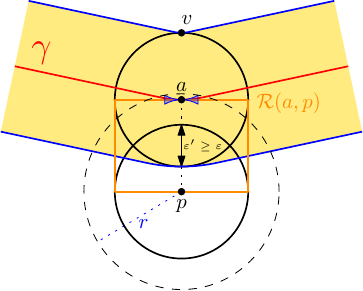}
    \caption{$p$ is $\eps$-blocking $\gamma$, $a$ is locally closest to $p$.}
    \label{fig:the_free_rectangle_is_free}
\end{figure}

Let $\gamma: [0, 1]\rightarrow \F$ be a geodesic path. We say that a point $a$ on $\gamma$ is \dfn{locally closest} to $p$ if there exists $\delta>0$ such that $D_{\|p-a\|}(p)\cap \gamma[a-\delta,a+\delta]=a$. 

\boundarySeparationRho*
\begin{proof} 
    First, because $\eps< \eps'=2-\|p-a\|$, we have $\sigma_\eps\ge \sigma_{\eps'}$, and therefore $\D_{\sigma_{\eps'}}(p)\subseteq \D_{\sigma_\eps}(p)$. 
    Let $v,u$ be the intersection points of the line through $p$ and $a$ with the boundary of $\D_1(a)$, and assume that $v$ is the point which is farther from $p$, as in \Cref{clm:free_rectangle} (see \Cref{fig:the_free_rectangle_is_free}). Then by \Cref{clm:free_rectangle}, $\R(a,p)$ is contained in $\D_2(v)\cup \D_{\sigma_\eps}(p)$. Because $\D_{\sigma_\eps}(p)\subset\W$ (\Cref{con:boundary_seperation_to_clear_R}), we have $\R(a,p)\cap \D_{\sigma_\eps}(p)\subset \W$.

    Now, since $a$ is a point on $\gamma$ which is locally closest to $p$, 
    the segment $\overline{uv}$ is the diametric chord of $\D_1(a)$ perpendicular to the tangent $\ell$ of $\gamma$ through $a$. Notice that by definition, we have $\R(a,p)\cap \D_2(v)\subset L(a,\overline{uv})$.
    To apply \Cref{clm:the_eye_is_in_the_trace} and get that $L(a,\overline{uv})\subset D_1(\gamma)\subset \W$, we first need to show that $\gamma(0),\gamma(1)\notin F_1(a,\overline{uv})\cup F_2(a,\overline{uv})$. This follows from \Cref{con:robot_seperation_to_ensure_R}: denote $r= \sqrt{\frac{1}{2}^{2} + \left(3-\frac{\sqrt{3}}{2}-\eps\right)^2}$. Because $\gamma(0),\gamma(1)\in S\cup T$, we get that $\|p-\gamma(1)\|\ge r$ and $\|p-\gamma(0)\|\ge r$, and thus $\gamma(0),\gamma(1)\notin S\cup D_r(p)$. One can easily verify that $D_r(p)$ contains the funnels $F_1(a,\overline{uv})\cup F_2(a,\overline{uv})$ (in fact, this is how we chose $r$, as the distance between $p$ and e.g. $z_3$, see \Cref{fig:coordinates}), and therefore $\gamma(0),\gamma(1)\notin F_1(a,\overline{uv})\cup F_2(a,\overline{uv})$.
    
    We conclude that $\R(a,p)\cup D_1(a)\cup D_1(p)\subset \W$, and therefore $\overline{pa}\subset \F$.
\end{proof}

\end{document}